\documentclass[12pt,onecolumn,draftcls,journal]{IEEEtran}
\ifCLASSINFOpdf
\else
\fi
\usepackage{mathrsfs}
\usepackage{caption2}
\usepackage{url}
\usepackage{stfloats}
\usepackage{fixltx2e}
\usepackage{eqparbox}
\usepackage{mdwmath}
\usepackage{mdwtab}
\usepackage{array}
\usepackage{algorithmic}
\usepackage[cmex10]{amsmath}
\usepackage{amsmath}
\usepackage{amsthm}

\usepackage{amssymb}
\usepackage{graphicx}
\usepackage{booktabs,multirow}
\usepackage{subfigure}
\begin{document}
%
\title{Performance Analysis of Interference-Limited Three-Phase Two-Way Relaying with \\ Direct Channel}
%
%
%

\author{Xiaochen~Xia,
        Dongmei~Zhang,
        Kui~Xu,~\IEEEmembership{Member,~IEEE},
        Youyun~Xu,~\IEEEmembership{Senior~Member,~IEEE}
\thanks{This work is supported by the Jiangsu Province Natural Science Foundation under Grant BK2011002, National Natural Science Foundation of China for Young Scholar (No. 61301165), National Natural Science Foundation of China (No. 61371123) and Jiangsu Province Natural Science Foundation for Young Scholar under Grant BK2012055. This work has been presented in part at the IEEE Wireless Communications and Networking Conference (WCNC), Shanghai, China, Apr. 2013.}
\thanks{X. Xia, D. Zhang, K. Xu and Y. Xu are with the Institute of Communication Engineering, PLA University of Science and Technology (email: tjuxiaochen@gmail.com, zhangdm72@163.com, lgdxxukui@gmail.com, yyxu@vip.sina.com). }}

\maketitle

\begin{abstract}

This paper investigates the performance of interference-limited three-phase two-way relaying with direct channel between two terminals in Rayleigh fading channels. The outage probability, sum bit error rate (BER) and ergodic sum rate are analyzed for a general model that both terminals and relay are corrupted by co-channel interference. We first derive the closed-form expressions of cumulative distribution function (CDF) for received signal-to-interference-plus-noise ratio (SINR) at the terminal. Based on the results for CDF, the lower bounds, approximate expressions as well as the asymptotic expressions for outage probability and sum BER are derived in closed-form with different computational complexities and accuracies. The approximate expression for ergodic sum rate is also presented. With the theoretic results, we consider the optimal power allocation at the relay and optimal relay location problems that aiming to minimize the outage and sum BER performances of the protocol. It is shown that jointly optimization of power and relay location can provide the best performance. Simulation results are presented to study the effect of system parameters while verify the theoretic analysis. The results show that three-phase TWR protocol can outperform two-phase TWR protocol in ergodic sum rate when the interference power at the relay is much larger than that at the terminals. This is in sharp contrast with the conclusion in interference free scenario. Moreover, we show that an estimation error on the interference channel will not affect the system performance significantly, while a very small estimation error on the desired channels can degrade the performance considerably.

\end{abstract}


%
\IEEEpeerreviewmaketitle

\section{Introduction}

In recent years, relaying has been accepted by several standards such as IEEE 802.11s, IEEE 802.16j and LTE-Advanced as a powerful technique to provide spatial diversity in cooperation systems and extend the coverage of the wireless networks. However, as shown in \cite{IEEEhowto:Kim2011}, the employment of relay doubles the required channels for transmission from source to destination due to the half-duplex constraint, which induces the spectral efficiency loss.

To improve the spectral efficiency, two-way relaying (TWR) or bi-directional relaying, which employs the idea of network coding (NC), has been investigated in \cite{IEEEhowto:Upadhyay}-\cite{IEEEhowto:Yi2009TWC}. In TWR, two terminals transmit their signals to a relay in one or two phases, and then the relay broadcasts the combination of the information extracted from the received signals. In \cite{IEEEhowto:Upadhyay},\cite{IEEEhowto:Kim_1}, the authors studied the two-phase TWR (2P-TWR) protocol with an amplify-and-forward (AF) relay. Wherein, the outage probability and diversity-multiplexing tradeoff have been analyzed. The performance and relay selection strategy of 2P-TWR protocol with multiple mobile relays were studied in \cite{IEEEhowto:Xiaochen2013IET}. The AF-based three-phase TWR (3P-TWR) protocol has been analyzed in \cite{IEEEhowto:Kim_3}, where the expression of outage probability has been obtained and the optimal power allocation scheme at the relay has been presented. In \cite{IEEEhowto:YadavCL2013}, the authors analyzed the performance of optimal relay selection for 3P-TWR protocol in Nakagami-$m$ fading channels and presented the closed-form expression for outage probability. In \cite{IEEEhowto:Yi2009TWC}, the authors showed that, in the interference free scenario, the 2P-TWR protocol outperforms 3P-TWR protocol in ergodic sum rate, while the 3P-TWR protocol performs better in outage and BER performances when the direct channel between two terminals exists.

In practical wireless network, signals of terminals (or relay) are often corrupted by co-channel interference (CCI) from other sources that share the same frequency resources in wireless networks \cite{IEEEhowto:Hoeher}. Moreover, CCI often dominates AWGN in wireless networks with dense frequency reuse. Therefore, it is necessary to take the effect of CCI into serious consideration in the analysis and design of the practical TWR protocol. Some of the previous studies have investigated the performance of TWR protocol in the interference-limited scenario. For example, the outage and BER performances of single terminal for two-phase AF-based TWR protocol have been analyzed in \cite{IEEEhowto:Ikki} for the interference-limited scenario. But this work considered only the special case that all interferers have the identical average interference power and the interference channels are independent identically distributed. In \cite{IEEEhowto:Liang2012}, the authors investigated the 2P-TWR in a more general scenario where interferers have different average interference powers. The expression of system outage probability \cite{IEEEhowto:Kim_1} was derived. In \cite{IEEEhowto:Anup2013}, the system outage performance of AF-based TWR protocol was analyzed using the a novel geometric method. Very recently, the effect of CCI was analyzed for TWR protocol in Nakagami-$m$ fading channels and the optimal resource allocation scheme was developed \cite{IEEEhowto:Soleimani2013TCOM}.

However, to the best of the authors knowledge, none of the aforementioned publications considered the performance of 3P-TWR protocol in the interference-limited scenario. The 3P-TWR protocol is suitable for the scenarios where the reliability has a higher priority in the system. As a result, it is of great importance to analyze the effect of CCI on the 3P-TWR protocol. Moreover, in this paper, we will show that the 3P-TWR protocol may outperform 2P-TWR protocol in ergodic sum rate when the effect of the CCI is taken into consideration. This contradicts with the conclusion obtained in the interference free scenario.

In this work, we study the performance of three-phase AF-based TWR protocol with direct channel between two terminals (This protocol is also called time division broadcasting protocol in \cite{IEEEhowto:Kim_3}) in the interference-limited scenario. A general model that all nodes (terminals and relay) are interfered by a finite number of co-channel interferers in the independent but non-identical Rayleigh fading channels is considered. The contributions of this paper are summarized as follows:

\begin{itemize}
\item The lower bounds for outage probability and sum bit error rate (BER) with infinite series are derived, which are shown to provide excellent estimation to the exact results obtained by simulation.
\item The approximate expressions without infinite series and asymptotic expressions for outage probability and sum BER are derived, which are tight in the low and high SNR regions, respectively. The approximate expression for ergodic sum rate is also obtained.
\item The optimal power allocation at the relay and optimal relay location, which aiming to minimize the outage and sum BER performances, are studied based on the asymptotic analysis.
\end{itemize}

The rest of this paper is organized as follow. In the next section, we will describe the system model and present the expression for the received signal-to-interference-plus-noise ratio (SINR) at terminal. The cumulative distribution function (CDF) of the received SINR at terminal is determined in section III. The outage, sum BER and ergodic sum rate performances are analyzed in section IV, section V and section VI, respectively. The optimal power allocation and relay location problems are studied in section VII. Simulation results are presented in section VIII and some conclusions will be drawn in the last section.

\section{System Model}
\begin{figure}[t]
\centering
\includegraphics[width=7.5cm]{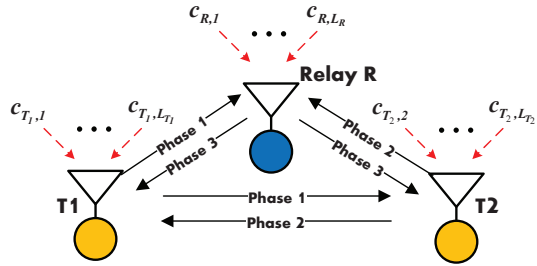}
\caption{The 3P-TWR protocol corrupted by a finite number of co-channel interferers.} \label{fig:graph}
\end{figure}

{\footnotetext[1]{In the scenarios where some interferers have multiple antennas, the channel coefficients of different interference channels may be correlated. However, this is beyond the scope of the current work and will be considered in the future study. Similar to \cite{IEEEhowto:Ikki}-\cite{IEEEhowto:Soleimani2013TCOM}, we assume the interferer has single antenna and the distances between the interferers are large enough. As a result, the channel coefficients of different interference channels are independent.}}

We consider the TWR network which consists of two terminals and a relay, as shown in Fig. 1, where terminals $T_1$ and $T_2$ wish to exchange information with the help of the relay $R$. Each node is equipped with a single antenna and operates in the half-duplex mode. It is assumed that both terminals and relay are interfered by a finite number of co-channel interferers. Here we let $L_R$, $L_{T_1}$ and $L_{T_2}$ denote the total numbers of interferers that affect node $R$, $T_1$ and $T_2$, respectively. Let $h_0$, $h_1$ and $h_2$ denote the channel coefficients between $T_1$ and $T_2$, $T_1$ and $R$, and $T_2$ and $R$ with variances $\Omega_0=d_{T_1,T_2}^{-v}$, $\Omega_1=d_{T_1,R}^{-v}$ and $\Omega_2=d_{T_2,R}^{-v}$, respectively, where $d_{U,N}$ denotes the distance between nodes $U$ and $N$. $v$ denotes the path loss exponent. Let $c_{N,k}\in{\cal CN}(0,\Omega_{N,k})$ denote the channel coefficient between node $N\in\{T_1,T_2,R\}$ and the $k$th interferer that affects $N$. All the channels are assumed to be reciprocal and {independent\footnotemark[1]} Rayleigh fading and the channel coefficients do not change within one round of data exchange.

One round of data exchange between two terminals can be achieved within three phases, i.e., $T_1$ transmits during the first phase, while $T_2$ and $R$ listen. In the second phase, $T_2$ transmits while $T_1$ and $R$ listen. The received signals at the relay during the first two phases can be expressed as
\begin{equation}
\begin{array}{l}
 y_R^{[1]} = \sqrt {{P_1}} {h_1}{S_1} + \sum_{k = 1}^{{L_R}} {\sqrt {{P_{I,R}}} {c_{R,k}}I_{R,k}^{[1]}}  + n_R^{[1]}, \\
 y_R^{[2]} = \sqrt {{P_2}} {h_2}{S_2} + \sum_{k = 1}^{{L_R}} {\sqrt {{P_{I,R}}} {c_{R,k}}I_{R,k}^{[2]}}  + n_R^{[2]}, \\
\end{array}
\end{equation}
where ${P_{I,N}}$ indicates the transmitted power of interferers that affect node $N\in\{T_1,T_2,R\}$. $S_i$ and $P_i$ ($i = 1,2,R$) denote the unit-power transmitted symbols and transmitted powers of nodes $T_1$, $T_2$ and $R$, respectively. $y_N^{[m]}$, ${\cal I}_{N,k}^{[m]}$ and $n_N^{[m]}\in{\cal {CN}}(0,1)$ represent the received signal, the unit-power interference signal of the $k$th interferer and the AWGN at node $N\in\{T_1,T_2,R\}$ during the $m$th phase, respectively, where $m\in\{1,2,3\}$. Meanwhile, the signals received by $T_1$ and $T_2$ during the first two phases can be written as
\begin{equation}
\begin{array}{l}
y_{{T_1}}^{[2]} = \sqrt {{P_{{2}}}} h_0{S_{{2}}} + \sum_{k = 1}^{{L_{T_1}}} {\sqrt {P_{I,{T_1}}} {c_{{T_1},k}}{\cal I}_{{T_1},k}^{[2]}}  + n_{{T_1}}^{[2]} \\
y_{{T_2}}^{[1]} = \sqrt {{P_{{1}}}} h_0{S_{{1}}} + \sum_{k = 1}^{{L_{T_2}}} {\sqrt {P_{I,{T_2}}} {c_{{T_2},k}}{\cal I}_{{T_2},k}^{[1]}}  + n_{{T_2}}^{[1]}
\end{array}
\end{equation}

In phase 3, $R$ broadcasts the combined information to $T_1$ and $T_2$. The combined signal can be written as $S_R \!=\! {\cal A}_1 y_R^{[1]} \!\!+\!\! {\cal A}_2 y_R^{[2]}$, where ${\cal A}_1$ and ${\cal A}_2$ denote the combining coefficients which can be determined {as\footnotemark[2]}
{\footnotetext[2]{Similar to \cite{IEEEhowto:Ikki} and \cite{IEEEhowto:Liang2012}, we assume $R$ knows the channel gains of links $T_1 \!\to\! R$, $T_2\!\to\! R$, and the instantaneous total interference power at $R$. Moreover, it is assumed $T_i$ knows the channel gains of links $T_1\!\to \!R$, $T_2\!\to\! R$, $T_1\!\to\! T_2$ and the instantaneous total interference powers at $R$ and $T_i$. The effect of channel state information (CSI) imperfection will be analyzed in simulations. Note that the performance based on the above assumptions can serve as a benchmark for other practical scenarios (e.g., the CSI estimation is imperfect).}}
\begin{equation}
{{\cal A}_i} = \sqrt {\frac{{{\omega _i}}}{{{\omega _1}{P_1}{{\left| {{h_1}} \right|}^2} + {\omega _2}{P_2}{{\left| {{h_2}} \right|}^2} + \sum_{k = 1}^{{L_R}} {{P_{I,R}}{{\left| {{c_{R,k}}} \right|}^2} + 1} }}}
\end{equation}
where $i\in \{1, 2\}$. $\omega _i\in \left(0,1\right)$ is the power allocation number adopted by the relay which satisfies ${\omega _1} + {\omega _2} = 1$. Then the received signal at $T_i$ during the third phase can be written as
\begin{equation}
\begin{array}{ll}
y_{{T_i}}^{[3]} = \sqrt {{P_R}} {h_i}{S_R} + \sum_{k = 1}^{{L_{{T_i}}}} {\sqrt{P_{I,{T_i}}}{c_{{T_i},k}}} {\cal I}_{{T_i},k}^{[3]} + n_{{T_i}}^{[3]}
\end{array}
\end{equation}

In the following, we assume equal power {allocation\footnotemark[3]} between $T_1$, $T_2$ and ${R}$, i.e., $P_1=P_2=P_R=P$. Since $T_i$ knows its own transmitted symbols, it can cancel the self-interference term in ${y_{T_i}^{[3]}}$. After performing maximal-ratio combining (MRC) on the received signals from direct channel and relay-to-terminal channel, the instantaneous SINR at $T_i$ can be tightly approximated as (See Appendix A)
{\footnotetext[3]{The assumption of equal power allocation does not make the analysis in this work lose generality because the variances of the channel coefficients between $T_1$, $T_2$ and $R$ can be different \cite{IEEEhowto:Ding},\cite{IEEEhowto:Lee}.}}
\begin{equation}
\Upsilon_{T_i} = {\Upsilon _{{T_i},D}} + \frac{{{\Upsilon _{{T_i},1}}{\Upsilon _{{T_i},2}}}}{{{\Upsilon _{{T_i},1}} + {\Upsilon _{{T_i},2}}}}
\end{equation}
where ${\Upsilon_{{T_i},D}} = \frac{{{\gamma _0}}}{{{\Gamma _{{T_i}}} + 1}}$ is the received SINR of channel $T_j\to T_i$. ${\Gamma _{{N}}}=\sum_{k=1}^{L_N}P_{I,N}\left|c_{N,k}\right|^2$ represents the total instantaneous interference power at node $N\in\{T_1,T_2,R\}$. $\gamma _i$ is defined as ${\gamma _i} \buildrel \Delta \over = P{\left| {{h_i}} \right|^2}$ with mean ${\bar \gamma _i} = {\mathbb E}({\gamma _i})$  and ${\mathbb E}(\cdot)$ indicates the expectation. $\Upsilon _{{T_i},1}$ and $\Upsilon _{{T_i},2}$ are given by
\begin{equation}
{\Upsilon _{{T_i},1}} = \frac{{{\gamma _i}}}{{{\Gamma _{{T_i}}} + 1}}, {\Upsilon _{{T_i},2}} = \frac{{{\omega _j}{\gamma _j}}}{{{\Gamma _R} + {\omega _i}{\Gamma _{{T_i}}} + {\omega _i} + 1}}
\end{equation}
where $i,j\in \{1,2\}$, $i\ne j$. Using the harmonic-to-min approximation, one can obtain the upper bound of $\Upsilon_{T_i}$, i.e.,
\begin{equation}
\Upsilon_{T_i}^{\rm{UB}} = {\Upsilon _{{T_i},D}} + \min \left\{ {{\Upsilon _{{T_i},1}},{\Upsilon _{{T_i},2}}} \right\}
\end{equation}

Note that MRC is suboptimal for the considered protocol in the interference-limited scenario. However, as shown in \cite{IEEEhowto:Shah2000},\cite{IEEEhowto:Chayawan2002}, the performance difference between MRC and optimal combining (OC) is not significant when the diversity branches are relatively small. As a result, we adopt MRC in this paper since the its performance is easier to analyze than OC, and meanwhile, it provides us a bound on the performance of OC.

\section{CDF of Received SINR at Terminal}
In this section, we derive the expressions of CDF for the received SINR at terminal. These will be used to derive the outage probability, sum BER and ergodic sum rate for the 3P-TWR protocol with CCI.

Attempting to derive the exact expression of CDF for received SINR at terminal in closed-form is challenging. Therefore, to make the analysis mathematically tractable, we consider the CDF for the upper bound derived in (7). Note that this CDF serves as a lower bound of that for exact received SINR at terminal. Without loss of generality, only the CDF of $\Upsilon_{T_1}^{\rm{UB}}$ will be derived and the CDF of $\Upsilon_{T_2}^{\rm{UB}}$ can be obtained vice versa. With the help of total probability theorem, the conditional CDF of $\Upsilon _{{T_1}}^{\rm{UB}}$ can be written as
\begin{equation}
\begin{aligned}
&   {F_{\gamma _{{T_1}}^{\rm{UB}}\left| {\left\{ {{\Gamma _R},{\Gamma _{{T_1}}}} \right\}} \right.}}\left( \gamma  \right) = \Pr \left( {\gamma _{{T_1}}^{\rm{UB}} < \gamma \left| {{\Gamma _R},{\Gamma _{{T_1}}}} \right.} \right)\\
&{} = 1 - \Pr \left( {\left. {{\Upsilon _{{T_1},D}} > \gamma } \right|{\Gamma _{{T_1}}}} \right) - \Pr \left( {\left. {{\Upsilon _{{T_1},D}} < \gamma ,{\Upsilon _{T_i}^m} > \gamma  - {\Upsilon _{{T_1},D}}} \right|{\Gamma _R},{\Gamma _{{T_1}}}} \right) \\
\end{aligned}
\end{equation}
The CDF of $\Upsilon _{{T_1}}^{\rm{UB}}$ can be obtained by averaging the conditional CDF with respect to the probability density functions (PDF) of $\Gamma_R$ and $\Gamma_{T_1}$, i.e.,
\begin{equation}
\begin{aligned}
 {F_{\Upsilon _{{T_1}}^{\rm{UB}}}}\left( \gamma  \right)  &= 1 - \int\limits_0^\infty  {{f_{{\Gamma _{{T_1}}}}}\left( t \right)\Pr \left( {\left. {{\Upsilon _{{T_1},D}} > \gamma } \right|{\Gamma _{T_1}}=t} \right)dt} \\
&{} - \int\limits_0^\infty  {\int\limits_0^\infty  {{f_{{\Gamma _R}}}\left( s \right)} } {f_{{\Gamma _{{T_1}}}}}\left( t \right) \Pr \left( {\left. {{\Upsilon _{{T_1},D}} < \gamma ,{\Upsilon _{T_i}^m} > \gamma  - {\Upsilon _{{T_1},D}}} \right|{\Gamma _R}=s,{\Gamma _{{T_1}}}=t} \right)dsdt\\
&{} \buildrel \Delta \over = 1 - {{\cal F}_1}\left( \gamma  \right) - {{\cal F}_2}\left( \gamma  \right)
\end{aligned}
\end{equation}
where ${{f_{X}}\left( x \right)}$ indicates the PDF of random variable (RV) $X$. Since ${\Gamma _N}$ ($N\in\{R,T_1,T_2\}$) is the sum of a finite number of exponential RVs with different means, the PDF of ${\Gamma _N}$ can be written {as\footnotemark[4]} \cite{IEEEhowto:Khuong}
{\footnotetext[4]{Herein, we will consider only the case $\xi_{N,i}\ne \xi_{N,k}$, $\forall i \ne k$. However, for the case $\exists  i \ne k$, $\xi_{N,i} = \xi_{N,k}$, the CDF can be derived in the similar way.}}
\begin{equation}
{f_{{\Gamma _N}}}\left( t \right) = \sum\limits_{k = 1}^{{L_N}} {{\phi _{N,k}}} \exp \left( { - \frac{1}{{{\xi_{N,k}}}}t} \right)
\end{equation}
where $\xi_{N,k}=P_{I,N}\Omega_{N,k}$ and ${\phi _{N,k}} = \prod_{i = 1,i \ne k}^{{L_N}} {\frac{1}{{{\xi_{N,k}} - {\xi_{N,i}}}}}$. Note that for the special case of $L_N=1$, ${f_{{\Gamma _N}}}\left( t \right)$ should be replaced by ${f_{{\Gamma _N}}}\left( t \right)= {\phi _{N,1}} \exp ( { - \frac{1}{{{\xi_{N,1}}}}t} )$, where $\phi _{N,1} = \frac{1}{\xi_{N,1}}$.

\newtheorem{lemma}{\emph{Lemma}}
\begin{lemma}\label{t1}
\emph{The closed-form expression of CDF for the SINR upper bound at $T_1$ can be expressed as}
\begin{equation}
\begin{array}{ll}
{F_{\Upsilon _{{T_1}}^{\rm{UB}}}}\left( \gamma  \right) = 1 -  \exp \left( { - \frac{\gamma }{{{{\bar \gamma }_0}}}} \right)\sum\limits_{j } {{\phi _{T_1,j}}} \frac{{{{\bar \gamma }_0}}}{{\gamma  + {{{{\bar \gamma }_0}} \mathord{\left/
{\vphantom {{{{\bar \gamma }_0}} {{\xi_{{T_1},j}}}}} \right.
\kern-\nulldelimiterspace} {{\xi_{{T_1},j}}}}}} \\
- \frac{{{\omega _2}}}{{{{\bar \gamma }_0}}}\sum\limits_j {\sum\limits_k {{\phi _{{T_1},j}}{\phi _{R,k}}} } \left( {{\rm{M}}\left( {1,1, - \frac{1}{{{{\bar \gamma }_2}}},\frac{1}{{{{\bar \gamma }_2}}},\frac{{{\omega _2}}}{{{\xi_{R,k}}}}} \right)} \right. + {\rm{M}}\left( {1,2, - \frac{1}{{{{\bar \gamma }_2}}},\frac{1}{{{{\bar \gamma }_2}}},\frac{{{\omega _2}}}{{{\xi_{R,k}}}}} \right) \\
+ {\rm{M}}\left( {\frac{1}{{{\omega _2}}},1,{\Phi _1},{\lambda _1},\frac{1}{{{\xi_{{T_1},j}}}}} \right)+ {\rm{M}}\left( {{\Phi _1}{{\bar \gamma }_2},2,{\Phi _1},{\lambda _1},\frac{1}{{{\xi_{{T_1},j}}}}} \right)\left. { + \Lambda \left( {{\lambda _1},{\lambda _2}} \right) - \Lambda \left( {\frac{1}{{{{\bar \gamma }_0}}},\frac{1}{{{{\bar \gamma }_0}}}} \right)} \right)
\end{array}
\end{equation}
\emph{where ${\Phi _i}$ and ${\lambda _i}$ are given by ${\Phi _i}= \frac{1}{{{{\bar \gamma }_0}}} - \frac{1}{{{{\bar \gamma }_1}}} - \frac{{{\omega _1} + i - 1}}{{{\omega _2}{{\bar \gamma }_2}}}$ and ${\lambda _i} = \frac{1}{{{{\bar \gamma }_1}}} + \frac{{{\omega _1}  + i - 1}}{{{\omega _2}}}\frac{1}{{{{\bar \gamma }_2}}}$, respectively. The functions ${\rm{M}}\left( {{\rho _1},{\rho _2},{\rho _3},{\rho _4},{\rho _5}} \right)$ and $\Lambda(\rho_1,\rho_2)$ are expressed {as\footnotemark[5]}}
{\footnotetext[5]{Unless explicitly stated, $\rm{y}$ and $\rm{y}(x_1,\cdots,x_n)$ will be used interchangeably to denote the function $\rm{y}(x_1,\cdots,x_n)$.}}
\begin{equation}
\begin{aligned}
{\rm{M}} &= \frac{{{\rho _1}\exp \left( { - {\lambda _2}\gamma } \right)}}{{{\rho _3}}}{\left( {\frac{1}{{{{\bar \gamma }_0}}}\gamma  + {\beta _{j,k}}} \right)^{ - {\rho _2}}} \sum\limits_{l = 0}^\infty  {{{\left( { - 1} \right)}^l}{{\left( {\frac{{{\rho _3}}}{{{\rho _4}\gamma  + {\rho _5}}}} \right)}^{l + 1}}\Phi _2^{ - l - 1}{\mathbb L}\left( {l + 1,{\Phi _2}\gamma } \right)}
\end{aligned}
\end{equation}
\emph{and}
\begin{equation}
{\Lambda} = \frac{{{{\bar \gamma }_2}}}{{\frac{1}{{{{\bar \gamma }_0}}}\gamma  + {\beta _{k,j}}}}\frac{1}{{{\rho _1}\gamma  + \frac{1}{{{\xi_{{T_1},j}}}}}}\exp \left( { - {\rho _2}\gamma } \right)
\end{equation}
\emph{where ${\mathbb L}(\cdot,\cdot)$ indicates the lower incomplete gamma function \cite{IEEEhowto:20}, ${\beta _{j,k}} = \frac{{{\omega _2}{{\bar \gamma }_2}{\Phi _1}}}{{{\xi_{R,k}}}} + \frac{1}{{{\xi_{{T_1},j}}}}$.}
\end{lemma}

\begin{proof}
See Appendix B.
\end{proof}

As shown in Lemma 1, since the expression of ${F_{\Upsilon _{{T_1}}^{\rm{UB}}}}\left( \gamma  \right)$ is given in a series form (introduced by ${\rm{M}}\left( {{\rho _1},{\rho _2},{\rho _3},{\rho _4},{\rho _5}} \right)$), more terms should be adopted in the calculation to obtain a higher accuracy, which leads to higher computational load. To alleviate the complexity, we present the approximate expression without infinite series and asymptotic expression of CDF for the SINR upper bound in the below.

\begin{lemma}\label{t1}
\emph{The approximate expression of CDF for ${\Upsilon _{{T_1}}^{\rm{UB}}}$ (denoted by ${F_{\Upsilon _{{T_1}}^{\rm{UB}}}^{\rm{App}}}\left( {{\gamma }} \right)$) can be obtained by replacing the function ${\rm{M}}\left( {{\rho _1},{\rho _2},{\rho _3},{\rho _4},{\rho _5}} \right)$ in Lemma 1 with ${\rm{\tilde M}}\left( {{\rho _1},{\rho _2},{\rho _3},{\rho _4},{\rho _5}} \right)$ which is given by}
\begin{equation}
{\rm{\tilde M}} = \frac{{{\rho _1}\exp \left( { - {\lambda _2}\gamma } \right)}}{{{\Phi _2}}}{\left( {\frac{1}{{{{\bar \gamma }_0}}}\gamma  + {\beta _{j,k}}} \right)^{ - {\rho _2}}}\frac{{1 - \exp \left( { - {\Phi _2}\gamma} \right)}}{{\left( {{\rho _3} + {\rho _4}} \right)\gamma  + {\rho _5}}}
\end{equation}
\end{lemma}

\begin{proof}
Recall (51) in Appendix B, it is shown that the integral ${\rm{M}}\left( {{\rho _1},{\rho _2},{\rho _3},{\rho _4},{\rho _5}} \right)$ can be approximated by
\begin{equation}
\begin{aligned}
{\rm{M}} \approx {\rho _1}&\exp \left( { - {\lambda _2}\gamma } \right){\left( {\frac{1}{{{{\bar \gamma }_0}}}\gamma  + {\beta _{j,k}}} \right)^{ - {\rho _2}}} \frac{1}{{\left( {{\rho _1} + {\rho _2}} \right)\gamma  + {\rho _3}}}\int\limits_0^\gamma  {\exp \left( { - {\Phi _2}z} \right)} dz
\end{aligned}
\end{equation}
Solving the integral in the above, one can arrived at (14).
\end{proof}

\begin{lemma}\label{t1}
\emph{The asymptotic expression of CDF for ${\Upsilon _{{T_1}}^{\rm{UB}}}$ (denoted by $F_{\Upsilon _{{T_1}}^{{\rm{UB}}}}^{{\rm{Asy}}}\left( \gamma  \right) $) can be expressed as}
\begin{equation}
\begin{aligned}
F_{\Upsilon _{{T_1}}^{{\rm{UB}}}}^{{\rm{Asy}}}\left( \gamma  \right) & = \frac{{{\gamma ^2}}}{{2{{\bar \gamma }_0}}}\left( {\sum\limits_j {\sum\limits_k {{\phi _{{T_1},j}}{\phi _{R,k}}\frac{{\xi _{R,k}^2}}{{{\omega _2}{{\bar \gamma }_2}}}\left( {{\xi _{{T_1},j}} + \xi _{{T_1},k}^2} \right)} } } \right.\\
&{} \left. { + \sum\limits_j {{\phi _{{T_1},j}}} \left( {{\lambda _2}{\xi _{{T_1},j}} + \left( {{\lambda _1} + {\lambda _2}} \right)\xi _{{T_1},j}^2 + 2{\lambda _1}\xi _{{T_1},j}^3} \right)} \right)
\end{aligned}
\end{equation}
\end{lemma}

\begin{proof}
See Appendix C.
\end{proof}

\section{Outage Probability Analysis}

\subsection{Lower Bound and Approximate Analysis}
In two-way relaying, there are two opposite traffic flows: one is from $T_1$ via $R$ to $T_2$, and the other is from $T_2$ via $R$ to $T_1$. So the \emph{system outage probability} is an important and commonly-used metric to evaluate the system performance \cite{IEEEhowto:Upadhyay}-\cite{IEEEhowto:Kim_3}. The system outage probability of 2P-TWR protocol can be efficiently derived by the geometric method proposed by \cite{IEEEhowto:Anup2013}. Unfortunately, the method can hardly be used in this paper since the non-outage probability \cite{IEEEhowto:Anup2013} for 3P-TWR protocol can not be expressed as in the form (or similar form) of [12, Eq. 9] due to the presence of direct channel.

To circumvent this obstacle, we first consider the definition of system outage probability. The system outage event occurs when the mutual information at either of the terminals falls below the target rate, or equivalently, the received SINR at either of the terminals is below the target SINR $\gamma_{th}$. Then the system outage probability for the 3P-TWR protocol $P_{\rm{sys}}^{\cal O}\left( {{\gamma _{th}}} \right)$ can be written as
\begin{equation}
\begin{aligned}
& P_{{\rm{sys}}}^{\cal O}\left( {{\gamma _{th}}} \right) = \Pr \left( {{\Upsilon _{{T_1}}} < {\gamma _{th}} \cup {\Upsilon _{{T_2}}} < {\gamma _{th}}} \right)\\
&{} = \Pr \left( {{\Upsilon _{{T_1}}} < {\gamma _{th}}} \right) + \Pr \left( {{\Upsilon _{{T_2}}} < {\gamma _{th}}} \right) - \Pr \left( {{\Upsilon _{{T_1}}} < {\gamma _{th}},{\Upsilon _{{T_2}}} < {\gamma _{th}}} \right)\\
&{} \approx P_{{T_1}}^{\cal O}\left( {{\gamma _{th}}} \right) + P_{{T_2}}^{\cal O}\left( {{\gamma _{th}}} \right) - P_{{T_1}}^{\cal O}\left( {{\gamma _{th}}} \right)P_{{T_2}}^{\cal O}\left( {{\gamma _{th}}} \right)
\end{aligned}
\end{equation}
where $P_{{T_i}}^{\cal O}\left( {{\gamma _{th}}} \right)$ denotes the outage probability at terminal $T_i$ with target SINR ${\gamma _{th}}$. The third step is obtained by assuming ${\Upsilon _{{T_1}}}$ and ${\Upsilon _{{T_2}}}$ are independent. As will be shown in the next subsection, the above approximation gives rise to an upper bound to the exact system outage probability when the transmitted power goes into infinity. In the below, we employ the following performance metric
\begin{equation}
P_{\rm{pro}}^{\cal O}(\gamma_{th}) \buildrel \Delta \over = P_{{T_1}}^{\cal O}\left( {{\gamma _{th}}} \right) + P_{{T_2}}^{\cal O}\left( {{\gamma _{th}}} \right) - P_{{T_1}}^{\cal O}\left( {{\gamma _{th}}} \right)P_{{T_2}}^{\cal O}\left( {{\gamma _{th}}} \right)
\end{equation}
which is called the protocol outage probability, to evaluate the system outage performance approximately, because this metric requires only the outage probability at single terminal. The tightness of the approximation is verified in Fig. 2, where the terminals and relay are placed in a straight line and the relay is set between $T_1$ and $T_2$. The normalized distance between $T_1$ and $T_2$ is set to one. The variance of $c_{N,k}$ ($N\in\{T_1,T_2,R\}$) is assumed to be evenly distributed on the interval $[0.1,1]$. It is shown that the protocol outage probability provides a good approximation to system outage probability especially in the moderate and high SNR regions ($P>8$dB). Although a gap can be observed in the low SNR region (see Fig. 2(b)), the result based on (18) is still much tighter than the results based on the asymptotic method in \cite{IEEEhowto:Ikki2013}. As a result, it is reasonable to employ the protocol outage probability in either performance analysis or practical implementation.

\begin{figure}[t]
  \centering
  \subfigure[log-log plot]{
    \label{fig:subfig:a} 
    \includegraphics[width=2.8in]{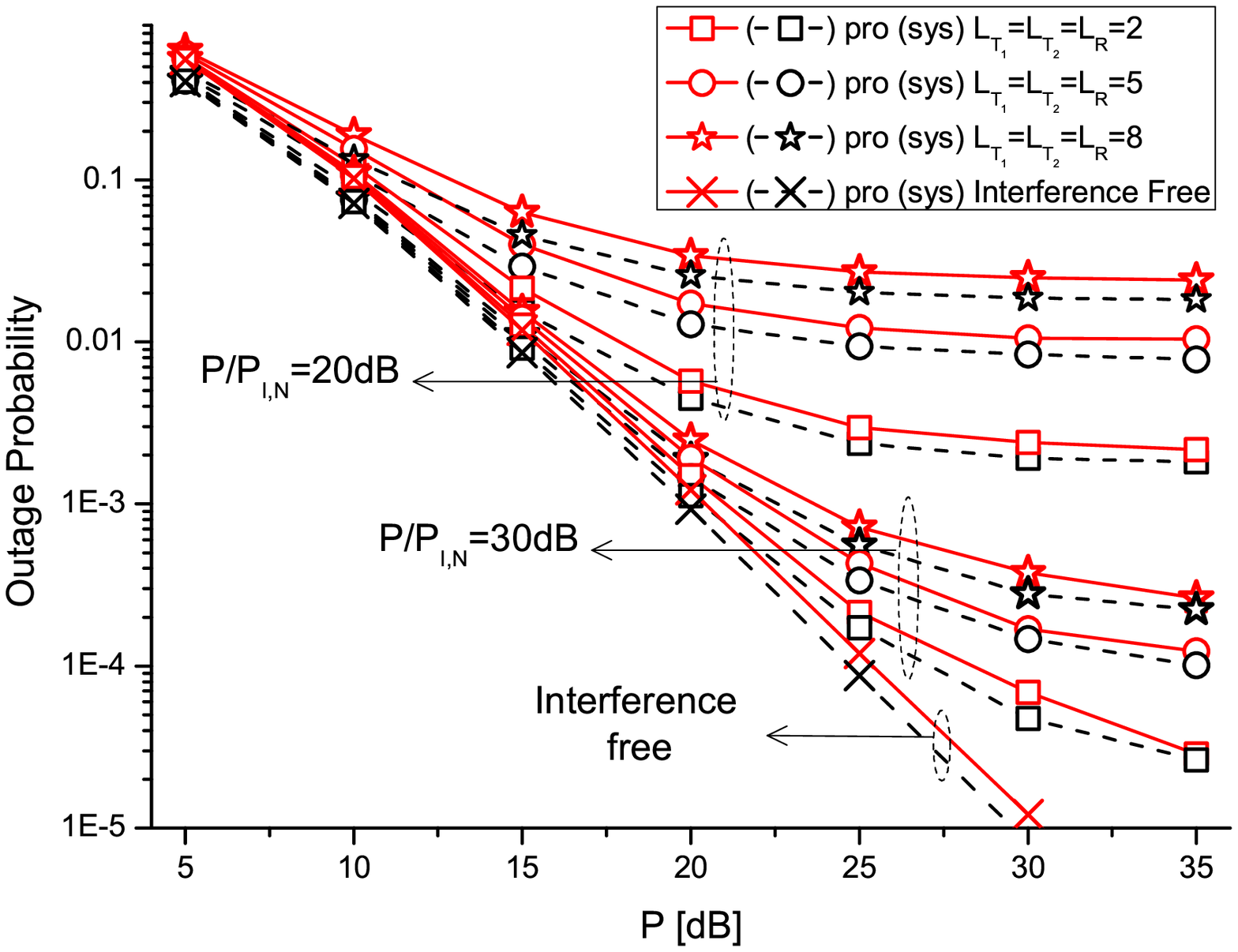}}
  \hspace{.1in}
  \subfigure[log-linear plot]{
    \label{fig:subfig:b} 
    \includegraphics[width=2.8in]{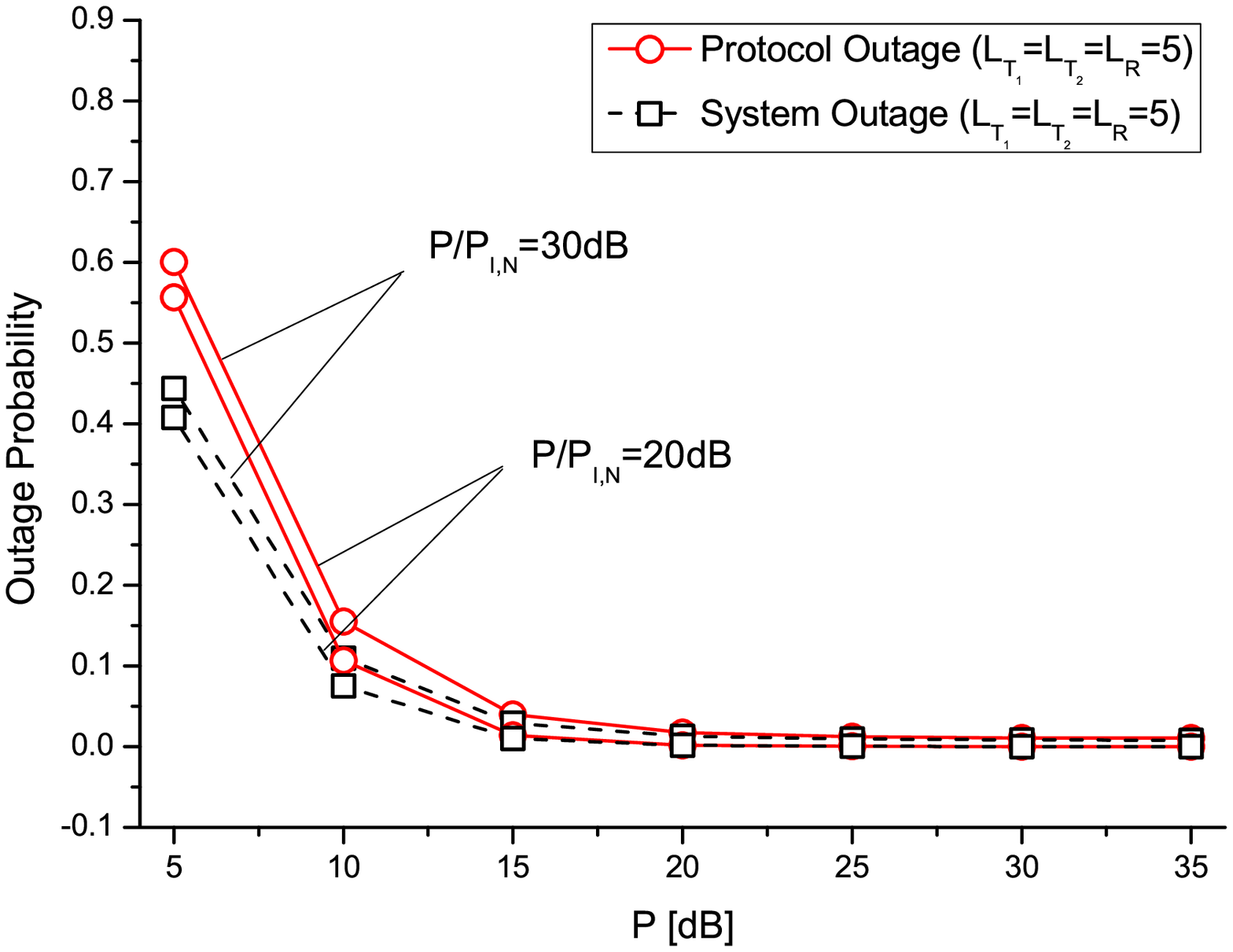}}
  \caption{Comparison between protocol outage probability and system outage probability, $\omega_1=\omega_2 = 0.5$, $d_{T_1,R}=d_{T_2,R}=0.5$, $\gamma_{th}=7$ (Corresponding to 1bit/s/Hz target rate). All the results are obtained by Monte-Carlo simulations.}
  \label{fig:subfig} 
\end{figure}



\newtheorem{theorem}{\emph{Theorem}}
\begin{theorem}\label{t1}
\emph{The lower bound and approximate expression of protocol outage probability, $P_{\rm{pro}}^{{\cal O},\rm{LB}}\left( {{\gamma _{th}}} \right)$ and $P_{\rm{pro}}^{{\cal O},\rm{App}}\left( {{\gamma _{th}}} \right)$, can be expressed as}
\begin{equation}
P_{\rm{pro}}^{{\cal O},\rm{LB}}\left( {{\gamma _{th}}} \right) = {F_{\Upsilon _{{T_1}}^{{\rm{UB}}}}}\left( {{\gamma _{th}}} \right) + {F_{\Upsilon _{{T_2}}^{{\rm{UB}}}}}\left( {{\gamma _{th}}} \right) - {F_{\Upsilon _{{T_1}}^{{\rm{UB}}}}}\left( {{\gamma _{th}}} \right){F_{\Upsilon _{{T_2}}^{{\rm{UB}}}}}\left( {{\gamma _{th}}} \right)
\end{equation}
\begin{equation}
P_{\rm{pro}}^{{\cal O},\rm{App}}\left( {{\gamma _{th}}} \right) = F_{\Upsilon _{{T_1}}^{{\rm{UB}}}}^{{\rm{App}}}\left( {{\gamma _{th}}} \right) + F_{\Upsilon _{{T_2}}^{{\rm{UB}}}}^{{\rm{App}}}\left( {{\gamma _{th}}} \right) - F_{\Upsilon _{{T_1}}^{{\rm{UB}}}}^{{\rm{App}}}\left( {{\gamma _{th}}} \right)F_{\Upsilon _{{T_2}}^{{\rm{UB}}}}^{{\rm{App}}}\left( {{\gamma _{th}}} \right)\
\end{equation}
\emph{where ${F_{\Upsilon _{{T_2}}^{\rm{UB}}}^{\rm{LB}}}\left( {{\gamma _{th}}} \right)$ and ${F_{\Upsilon _{{T_2}}^{\rm{UB}}}^{\rm{App}}}\left( {{\gamma _{th}}} \right)$ denote the lower bound and approximate expression of CDF for ${\Upsilon _{{T_2}}^{\rm{UB}}}$, respectively.}
\end{theorem}


\begin{proof}
The proof is straightforward according to (18), and thus it is neglect.
\end{proof}

Note that the approximate expression of protocol outage probability does not require the computation of infinite series according to Lemma 2.
\subsection{Asymptotic Analysis}

To get more insight about the effect of system parameters on the protocol outage probability, we provide asymptotic analysis based on the result in Lemma 3.

\begin{theorem}\label{t1}
\emph{The asymptotic expression of protocol outage probability $P_{\rm{pro}}^{{\cal O},\rm{Asy}}\left( {{\gamma _{th}}} \right)$ can be expressed as
\begin{equation}
\begin{aligned}
&   P_{\rm{pro}}^{{\cal O},\rm{Asy}}\left( {{\gamma _{th}}} \right)  = {F_{\Upsilon _{{T_1}}^{\rm{UB}}}^{\rm{Asy}}}\left( {{\gamma _{th}}} \right) + {F_{\Upsilon _{{T_2}}^{\rm{UB}}}^{\rm{Asy}}}\left( {{\gamma _{th}}} \right)\\
&{} = \sum\limits_{i = 1}^2 {\frac{{{\gamma ^2}}}{{2{{\bar \gamma }_0}}}\left( {\frac{{{\omega _i}\left( {{{\Gamma ''}_{{T_i}}} + 2{{\Gamma '}_{{T_i}}} + 1} \right) + \left( {{{\Gamma '}_R} + 1} \right)\left( {{{\Gamma '}_{{T_i}}} + 1} \right)}}{{{\omega _j}{{\bar \gamma }_j}}} + \frac{{{{\Gamma ''}_{{T_i}}} + 2{{\Gamma '}_{{T_i}}} + 1}}{{{{\bar \gamma }_i}}}} \right)}
\end{aligned}
\end{equation}
where $j\in{1,2}$ and $j\ne i$. ${{\Gamma '}_N} = \mathbb E\left( {{\Gamma _N}} \right)$ denotes the average received interference power at node $N$ and ${{\Gamma ''}_N} =  \mathbb E\left( {\Gamma _N^2} \right)$.}
\end{theorem}

\begin{proof}
We first note that $F_{\Upsilon _{{T_1}}^{{\rm{UB}}}}^{{\rm{Asy}}}\left( \gamma  \right)\times F_{\Upsilon _{{T_2}}^{{\rm{UB}}}}^{{\rm{Asy}}}\left( \gamma  \right)$ is the infinitesimal of higher order of $F_{\Upsilon _{{T_i}}^{{\rm{UB}}}}^{{\rm{Asy}}}\left( \gamma  \right)$ when $\gamma \to 0$ according to Lemma 3, as a result, the asymptotic expression can be written as in the first line of (21). Furthermore, using the relations $\sum\nolimits_k {{\phi _{N,k}}{\xi _{N,k}}}  = \int_0^\infty  {{f_{{\Gamma _N}}}\left( x \right)dx = 1}$, $\sum\nolimits_k {{\phi _{N,k}}\xi _{N,k}^2}  = \int_0^\infty  {x{f_{{\Gamma _N}}}\left( x \right)dx}  = {{\Gamma '}_N}$ and $\sum\nolimits_k {{\phi _{N,k}}\xi _{N,k}^3}  = \int_0^\infty  {{x^2}{f_{{\Gamma _N}}}\left( x \right)dx}  = {{\Gamma ''}_N}$, one can obtain the second line of (21).
\end{proof}

According to Theorem 2 and the second line of (17), we can see that the protocol outage probability serves as an upper bound of the system outage probability when $\gamma_{th} \to 0$ (or equivalently, $P \to \infty$ \cite{IEEEhowto:Wang}). Moreover, from Theorem 2, it is clear that the protocol outage probability increases as the average received interference powers at $T_1$, $T_2$ and $R$ increasing and decreases as the useful power increasing. Moreover, since we have ${\Gamma'' _N}  \propto {P^2_{I,N}}$, ${\Gamma' _N}  \propto {P_{I,N}}$ and ${\bar \gamma _i} \propto P$ ($i\in\{0,1,2\}$), it is easy to show that the protocol outage probability is proportional to a constant when the ratio of interference power ${P_{I,N}}$ to useful power $P$ is fixed. This indicates that the achievable diversity (defined as $d = -\mathop {\lim }\limits_{P \to \infty } [ {{\log \left( {P_{{\rm{sys}}}^{\cal O}\left( {{\gamma _{th}}} \right)} \right)} \mathord{\left/
 {\vphantom {{\log \left( {P_{{\rm{pro}}}^{\cal O}\left( {{\gamma _{th}}} \right)} \right)} {\log P}}} \right.
 \kern-\nulldelimiterspace} {\log P}} ]$ \cite{IEEEhowto:Suraweera2011ICC},\cite{IEEEhowto:Salhab2013CL}) of 3P-TWR protocol is zero in the interference-limited scenario.

\section{Sum Bit Error Rate Analysis}
In this section, we consider the sum BER performance which is defined as the sum of two terminals' average BERs \cite{IEEEhowto:Louie2010}.


\subsection{Lower Bound and Approximate Analysis}
We first derive the lower bound of sum BER based on Lemma 1. According to \cite{IEEEhowto:Chen2004}, the sum BER for several types of modulations employed in practical systems can be expressed as a function of CDFs for the received SINRs at two terminals. As a result, the lower bound of sum BER $P_{\rm{sum}}^{{\cal E},\rm{LB}}$ can be expressed as
\begin{equation}
P_{\rm{sum}}^{{\cal E},\rm{LB}} = P_{{T_1}}^{{\cal E},\rm{LB}} + P_{{T_2}}^{{\cal E},\rm{LB}} = a\sqrt {\frac{b}{\pi }} \int\limits_0^\infty  {{F_{\Upsilon _{{T_1}}^{\rm{UB}}}}\left( \gamma  \right)\frac{{\exp \left( { - b\gamma } \right)}}{{\sqrt \gamma  }}} d\gamma + a\sqrt {\frac{b}{\pi }} \int\limits_0^\infty  {{F_{\Upsilon _{{T_2}}^{\rm{UB}}}}\left( \gamma  \right)\frac{{\exp \left( { - b\gamma } \right)}}{{\sqrt \gamma  }}} d\gamma
\end{equation}
where $P_{{T_i}}^{{\cal E},\rm{LB}}$ indicates the lower bound of average BER at terminal $T_i$. $a$ and $b$ are modulation-related constants. For example, we have $(a,b)=(0.5,1)$ for BPSK modulation and $(a,b)=(0.5,0.5)$ for QPSK modulation. Due to the symmetry, we provide only the derivation of $P_{{T_1}}^{{\cal E},\rm{LB}}$.

\begin{theorem}\label{t1}
\emph{The lower bound of average BER at $T_1$ (denoted by $P_{{T_1}}^{{\cal E},\rm{LB}}$) can be expressed in closed-form as
\begin{equation}
\begin{array}{ll}
P_{{T_1}}^{{\cal E},\rm{LB}} = a \!-\! {{\bar \gamma }_0}\sum\limits_{j } {{\phi _{T_1,j}}} \frac{{a\sqrt b }}{{\sqrt {{{{{\bar \gamma }_0}} \mathord{\left/
 {\vphantom {{{{\bar \gamma }_0}} {{\xi_{{T_1},j}}}}} \right.
 \kern-\nulldelimiterspace} {{\xi_{{T_1},j}}}}} }}{\mathbb H}\left( {\frac{1}{2},\frac{1}{2},\frac{{1 + {{\bar \gamma }_0}b}}{{{\xi_{{T_1},j}}}}} \right) \\
- \frac{{{\omega _2}}}{{{{\bar \gamma }_0}}}\sum\limits_j {\sum\limits_k {{\phi _{T_1,j}}{\phi _{R,k}}} } \left( {{\rm M}^{\cal E}}\left( {1,1, - \frac{1}{{{{\bar \gamma }_2}}},\frac{1}{{{{\bar \gamma }_2}}},\frac{{{\omega _2}}}{{{\xi_{R,k}}}}} \right) + {{\rm M}^{\cal E}}\left( {1,2, - \frac{1}{{{{\bar \gamma }_2}}},\frac{1}{{{{\bar \gamma }_2}}},\frac{{{\omega _2}}}{{{\xi_{R,k}}}}} \right) \right.\\
+ {{\rm M}^{\cal E}}\left( {\frac{1}{{{\omega _2}}},1,{\Phi _1},{\lambda _1},\frac{1}{{{\xi_{{T_1},j}}}}} \right)+\! {{\rm M}^{\cal E}}\!\left( {{\Phi _1}{{\bar \gamma }_2},2,{\Phi _1},{\lambda _1},\frac{1}{{{\xi_{{T_1},j}}}}} \right)\left. { \!+\! {\Lambda ^{\cal E}}\!\left( {{\lambda _1},{\lambda _2}} \right) \!-\! {\Lambda ^{\cal E}}\!\left( {\frac{1}{{{{\bar \gamma }_0}}},\frac{1}{{{{\bar \gamma }_0}}}} \right)} \right)
\end{array}
\end{equation}
The expressions of ${\rm M}^{{\cal E}}(\rho_1,\rho_2,\rho_3,\rho_4,\rho_5)$ and ${\Lambda ^{\cal E}}(\rho_1,\rho_2)$ are given in Table I, where $G(\cdot)$ and ${\mathbb H}(\cdot,\cdot,\cdot)$ indicate the gamma function and confluent hypergeometric function of the second kind \cite{IEEEhowto:20}, respectively. ${\mu _i}$ and ${\nu _i}$ in ${\rm M}^{{\cal E}}(\rho_1,\rho_2,\rho_3,\rho_4,\rho_5)$ are expressed as
\begin{equation}
\begin{aligned}
& {\mu _i} = \frac{1}{{\left( {l - i + 1} \right)!\rho _4^{l - i + 1}}}\left.  {{{d^{l - i + 1}}{{\left( {\frac{1}{{{{\bar \gamma }_0}}}\gamma  + {\beta _{j,k}}} \right)}^{ - {\rho _2}}}} \mathord{\left/
 {\vphantom {{{d^{l - i + 1}}{{\left( {\frac{1}{{{{\bar \gamma }_0}}}\gamma  + {\beta _{j,k}}} \right)}^{ - {\rho _2}}}} {d{\gamma ^{l - i + 1}}}}} \right.
 \kern-\nulldelimiterspace} {d{\gamma ^{l - i + 1}}}} \right|_{\gamma  =  - {{{\rho _5}} \mathord{\left/
 {\vphantom {{{\rho _5}} {{\rho _4}}}} \right.
 \kern-\nulldelimiterspace} {{\rho _4}}}} \\
&{} {\nu _i} = \frac{1}{{\left( {{\rho _2} - i} \right)!{{\bar \gamma }_0}^{i - {\rho _2}}}}{\left. {{{{d^{{\rho _2} - i}}\frac{1}{{{{\left( {{\rho _4}\gamma  + {\rho _5}} \right)}^{l + 1}}}}} \mathord{\left/
 {\vphantom {{{d^{{\rho _2} - i}}\frac{1}{{{{\left( {{\rho _4}\gamma  + {\rho _5}} \right)}^{l + 1}}}}} {d{\gamma ^{{\rho _2} - i}}}}} \right.
 \kern-\nulldelimiterspace} {d{\gamma ^{{\rho _2} - i}}}}} \right|_{\gamma  =  - {{\bar \gamma }_0}{\beta _{j,k}}}}
\end{aligned}
\end{equation}}
\end{theorem}

\begin{proof}
See Appendix D.
\end{proof}

Although the lower bound derived in the above provides high accuracy as will be shown in the simulations, its practical applications are limited by the double infinite series introduced by ${\rm M}^{\cal E}\left( {{\rho _1},{\rho _2},{\rho _3},{\rho _4},{\rho _5}} \right)$. To alleviate the complexity, we then derived the approximate expression for BER at $T_1$ which does not require the computation of infinite series.

\begin{table}[!t]
\caption{Some Useful Functions}
\label{table_example}
\centering
\begin{tabular}{lll}
\includegraphics[width=15.6cm]{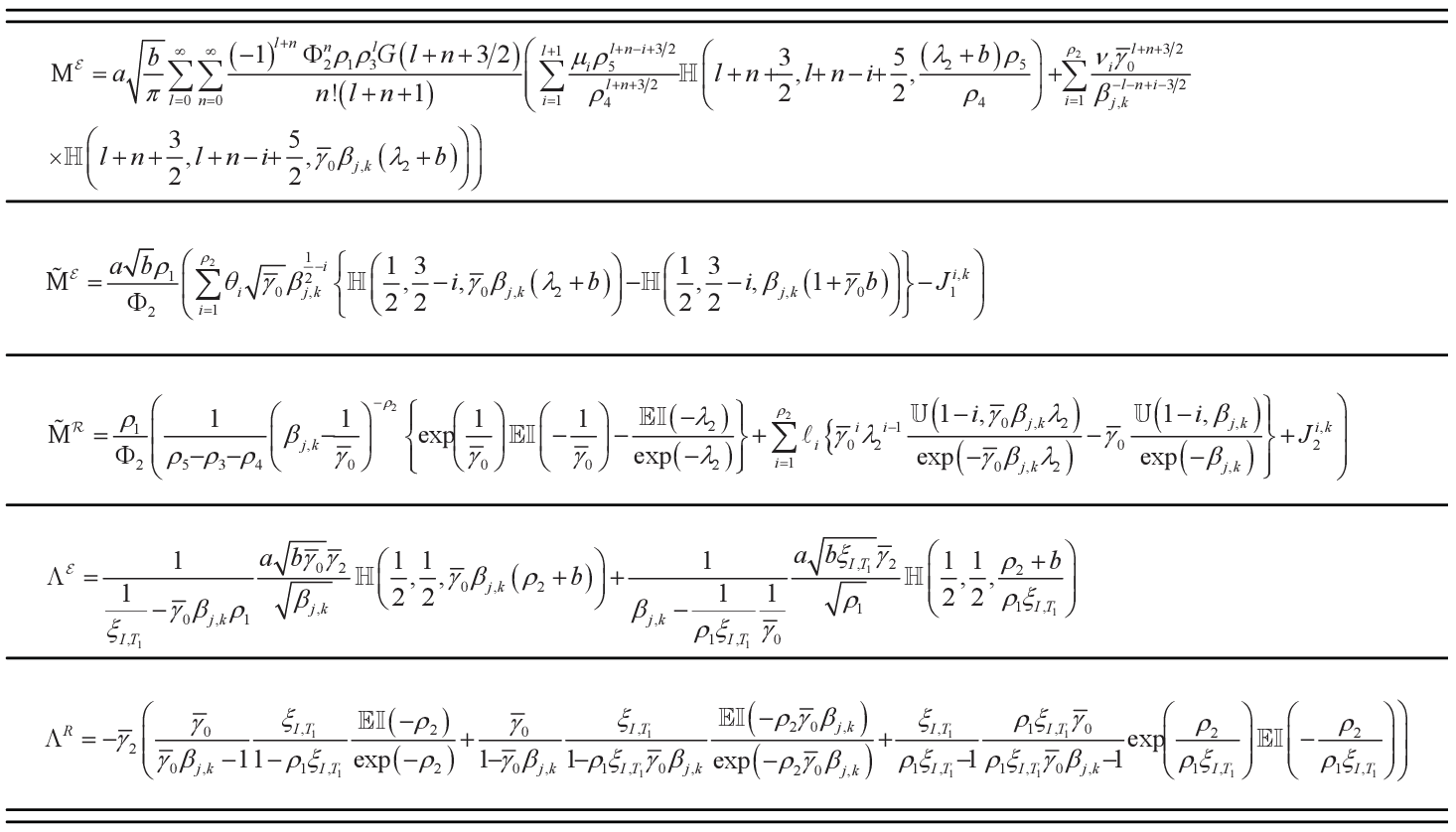}
\end{tabular}
\end{table}

\begin{theorem}\label{t1}
\emph{The approximate expression for average BER at $T_1$ (denoted by $P_{{T_1}}^{{\cal E},\rm{App}}$) can be obtained by replacing the function ${\rm{M}^{\cal E}}\left( {{\rho _1},{\rho _2},{\rho _3},{\rho _4},{\rho _5}} \right)$ in Theorem 3 with ${\rm{\tilde M}^{\cal E}}\left( {{\rho _1},{\rho _2},{\rho _3},{\rho _4},{\rho _5}} \right)$ which is given in Table I, where ${\theta_i}$ in ${\rm{\tilde M}^{\cal E}}\left( {{\rho _1},{\rho _2},{\rho _3},{\rho _4},{\rho _5}} \right)$ is given by
\begin{equation}
{\left. {{\theta _i} = \frac{1}{{\bar \gamma _0^{i - {\rho _2}}\left( {{\rho _2} - i} \right)!}}{{{d^{{\rho _2} - i}}\left[ {\frac{1}{{1 + \gamma }}\frac{1}{{\left( {{\rho _3} + {\rho _4}} \right)\gamma  + {\rho _5}}}} \right]} \mathord{\left/
 {\vphantom {{{d^{{\rho _2} - i}}\left[ {\frac{1}{{1 + \gamma }}\frac{1}{{\left( {{\rho _3} + {\rho _4}} \right)\gamma  + {\rho _5}}}} \right]} {d{\gamma ^{{\rho _2} - i}}}}} \right.
 \kern-\nulldelimiterspace} {d{\gamma ^{{\rho _2} - i}}}}} \right|_{\gamma  =  - {{\bar \gamma }_0}{\beta _{j,k}}}}
\end{equation}
$J_1^{j,k}$ in ${\rm{\tilde M}^{\cal E}}\left( {{\rho _1},{\rho _2},{\rho _3},{\rho _4},{\rho _5}} \right)$ is expressed as
\begin{equation}
J_1^{j,k} = \left\{ \begin{array}{l}
0,{\rho _3} + {\rho _4} = 0\\
{\left( {{\beta _{j,k}} - \frac{\rho }{{{{\bar \gamma }_0}}}} \right)^{ - {\rho _2}}}\frac{{\rho _5^{ - \frac{1}{2}}}}{{\sqrt {{\rho _3} + {\rho _4}} }}\left\{ {\mathbb H\left( {\frac{1}{2},\frac{1}{2},\rho \left( {{\lambda _2} + b} \right)} \right) - \mathbb H\left( {\frac{1}{2},\frac{1}{2},\left( {\frac{1}{{{{\bar \gamma }_0}}} + b} \right)\rho } \right)} \right\},{\rho _3} + {\rho _4} \ne 0
\end{array} \right.
\end{equation}
where $\rho  = \frac{{{\rho _5}}}{{{\rho _3} + {\rho _4}}}$.}
\end{theorem}


\begin{proof}
According to (22) and Lemma 2, the approximate expression can be obtained by replacing ${{\rm{M}}^{\cal E}}\left( {{\rho _1},{\rho _2},{\rho _3},{\rho _4},{\rho _5}} \right)$ with integral ${{\rm{\tilde M}}^{\cal E}}\left( {{\rho _1},{\rho _2},{\rho _3},{\rho _4},{\rho _5}} \right) = a\sqrt {\frac{b}{\pi }} \int_0^\infty  {\frac{{\exp \left( { - b\gamma } \right)}}{{\sqrt \gamma  }}} {\rm{\tilde M}}\left( {{\rho _1},{\rho _2},{\rho _3},{\rho _4},{\rho _5}} \right)d\gamma $, where ${\rm{\tilde M}}\left( {{\rho _1},{\rho _2},{\rho _3},{\rho _4},{\rho _5}} \right)$ is given by (14). Performing {partial\footnotemark[6]} fraction on term ${\left( {\frac{1}{{{{\bar \gamma }_0}}}\gamma  + {\beta _{j,k}}} \right)^{ - {\rho _2}}}\frac{{{\rho _1}}}{{\left( {{\rho _3} + {\rho _4}} \right)\gamma  + {\rho _5}}}$ and using [28, 9.211.4] on the resultant expression, ${{\rm{\tilde M}}^{\cal E}}\left( {{\rho _1},{\rho _2},{\rho _3},{\rho _4},{\rho _5}} \right)$ can be expressed as the second line of Table I.
\end{proof}
{\footnotetext[6]{The partial fraction will be frequently used in the derivations of this paper. Therefore, for the partial fraction of $\frac{1}{{{{\left( {{a_1}x + {b_1}} \right)}^{{m_1}}} \cdots {{\left( {{a_n}x + {b_n}} \right)}^{{m_n}}}}}$, we only consider the case $\frac{b_i}{a_i}\ne\frac{b_j}{a_j}$, $\forall i,j\in\{1,\cdots,n\}$ and $i \ne j$. For the special case $\exists i \ne j$, $\frac{b_i}{a_i} = \frac{b_j}{a_j}$, the results can be obtained by using the similar method, thus is neglect for the sake of clarity.}}

Note that the expressions in Theorem 3 and Theorem 4 only involve the special function ${\mathbb H}(\cdot,\cdot,\cdot)$, which can be easily evaluated by softwares such as Mathematica and Matlab. Moreover, Theorem 3 and Theorem 4 present only the expressions of BER for single terminal for the sake of clarity. The expressions for sum BER can be simply derived by exploiting the symmetry between two terminals.

\subsection{Asymptotic Analysis}
Next, we present the asymptotic expression which allows us to fast estimate the sum BER performance in the high SNR region. Inserting (16) into (22) and employing the integral result reported in [28, 3.381.4], we can obtain the following theorem.

\begin{theorem}\label{t1}
\emph{The asymptotic expression of sum BER (denoted by $P_{{\rm{sum}}}^{{\cal E},{\rm{Asy}}}$) can be expressed as
\begin{equation}
P_{{\rm{sum}}}^{{\cal E},{\rm{Asy}}} = \frac{a}{{\sqrt \pi  {b^2}\gamma _{th}^2}}G\left( {\frac{5}{2}} \right)P_{{\rm{pro}}}^{{\cal O},{\rm{Asy}}}\left( {{\gamma _{th}}} \right)
\end{equation}}
\end{theorem}

Note that the asymptotic sum BER is independent with the target SINR ${\gamma _{th}}$ because $P_{{\rm{pro}}}^{{\cal O},{\rm{Asy}}}\left( {{\gamma _{th}}} \right) \propto {\gamma ^2_{th}}$. Theorem 5 shows that the sum BER is a linear function of the protocol outage probability defined by (18) when the transmitted power goes into infinity. As a result, the asymptotic behavior for the sum BER is similar with that for protocol outage probability analyzed in section IV-B.

\section{Ergodic Sum Rate Analysis}
Another important metric to evaluate the system performance is ergodic sum rate which is defined as the sum of two terminals' average achievable rates. For 3P-TWR protocol, the ergodic sum rate can be expressed as
\begin{equation}
{R_{\rm{sum}}} \le {\mathbb E}\left[\frac{1}{3} {{{\log }_2}\left( {1 + \gamma _{{T_1}}^{\rm{UB}}} \right)} \right] + {\mathbb E}\left[ \frac{1}{3} {{{\log }_2}\left( {1 + \gamma _{{T_2}}^{\rm{UB}}} \right)} \right]
\end{equation}
where the pre-log factor of 1/3 is due to the fact that three phases are required for one round of data exchange between two terminals. Also, we consider only the average achievable rate for $T_1$ and derive the approximate expression based on Lemma 2 provided in section III. As shown in \cite{IEEEhowto:Al-Qahtani}, the average achievable rate at $T_1$ can also be expressed as a function of CDF for the received SINR at terminal, i.e.,
\begin{equation}
R_{{T_1}}^{{\rm{App}}} = \frac{1}{\tau }\int\limits_0^\infty  {\frac{{1 - F_{\Upsilon _{{T_1}}^{\rm{UB}}}^{{\rm{App}}}\left( \gamma  \right)}}{{1 + \gamma }}} d\gamma
\end{equation}
where $\tau = 3\ln(2)$.

\begin{theorem}\label{t1}
\emph{The approximate expression of average achievable rate at $T_1$ can be expressed in closed-form as
\begin{equation}
\begin{array}{ll}
R_{{T_1}}^{{\rm{App}}} = \frac{{{{\bar \gamma }_0}}}{\tau }\sum\limits_{j } {{\phi _{T_1,j}}{\chi _j}} \left( {\exp \left( {\frac{1}{{{{\bar \gamma }_0}}}} \right)\mathbb {EI}\left( { - \frac{1}{{{{\bar \gamma }_0}}}} \right) - \exp \left( {\frac{1}{{{\xi_{{T_1},j}}}}} \right)\mathbb {EI}\left( { - \frac{1}{{{\xi_{{T_1},j}}}}} \right)} \right)+ \frac{{{\omega _2}}}{{{{\bar \gamma }_0}\tau}}\sum\limits_j {\sum\limits_k {{\phi _{T_1,j}}{\phi _{R,k}}} }\\
\times \left\{ {{\tilde{\rm M}^{\cal R}}\left( {1,1, - \frac{1}{{{{\bar \gamma }_2}}},\frac{1}{{{{\bar \gamma }_2}}},\frac{{{\omega _2}}}{{{\xi_{R,k}}}}} \right) + {\tilde{\rm M}^{\cal R}}\left( {1,2, - \frac{1}{{{{\bar \gamma }_2}}},\frac{1}{{{{\bar \gamma }_2}}},\frac{{{\omega _2}}}{{{\xi_{R,k}}}}} \right)}+ {\tilde{\rm M}^{\cal R}}\left( {\frac{1}{{{\omega _2}}},1,{\Phi _1},{\lambda _1},\frac{1}{{{\xi_{{T_1},j}}}}} \right) \right.\\
+ {\tilde{\rm M}^{\cal R}}\left( {{\Phi _1}{{\bar \gamma }_2},2,{\Phi _1},{\lambda _1},\frac{1}{{{\xi_{{T_1},j}}}}} \right)\left. { + {\Lambda ^{\cal R}}\left( {{\lambda _1},{\lambda _2}} \right) - {\Lambda ^{\cal R}}\left( {\frac{1}{{{{\bar \gamma }_0}}},\frac{1}{{{{\bar \gamma }_0}}}} \right)} \right\}
\end{array}
\end{equation}
where ${\chi _j} = \frac{1}{{1 - {{{{\bar \gamma }_0}} \mathord{\left/ {\vphantom {{{{\bar \gamma }_0}} {{\xi_{{T_1},j}}}}} \right. \kern-\nulldelimiterspace} {{\xi_{{T_1},j}}}}}}$. The function ${\rm \tilde M}^{\cal R}\left( {{\rho _1},{\rho _2},{\rho _3},{\rho _4},{\rho _5}} \right)$ and ${\Lambda ^{\cal R}}\left( {{\rho _1},{\rho _2}} \right)$ are given by Table I. $\mathbb {EI}(\cdot)$ and $\mathbb U(\cdot,\cdot)$ denote the exponential integral function and upper incomplete gamma function \cite{IEEEhowto:20}, respectively. ${\ell _i}$ in ${\rm \tilde M}^{\cal R}\left( {{\rho _1},{\rho _2},{\rho _3},{\rho _4},{\rho _5}} \right)$ is given by
\begin{equation}
{\ell _i} = \frac{1}{{\bar \gamma _0^{i - {\rho _2}}\left( {{\rho _2} - i} \right)!}} \left. {{{d^{{\rho _2} - i}}\left( {\frac{1}{{1 + \gamma }}\frac{1}{{\left( {{\rho _3} + {\rho _4}} \right)\gamma  + {\rho _5}}}} \right)} \mathord{\left/
 {\vphantom {{{d^{{\rho _2} - i}}\left( {\frac{1}{{1 + \gamma }}\frac{1}{{\left( {{\rho _3} + {\rho _4}} \right)\gamma  + {\rho _5}}}} \right)} {d{\gamma ^{{\rho _2} - i}}}}} \right.
 \kern-\nulldelimiterspace} {d{\gamma ^{{\rho _2} - i}}}} \right|_{\gamma  =  - {{\bar \gamma }_0}{\beta _{j,k}}}
\end{equation}
$J_2^{j,k}$ in ${\rm \tilde M}^{\cal R}\left( {{\rho _1},{\rho _2},{\rho _3},{\rho _4},{\rho _5}} \right)$ can be expressed as
\begin{equation}
J_2^{j,k} = \left\{ {\begin{array}{*{20}{l}}
{0,{\rho _3} + {\rho _4} = 0}\\
{\frac{1}{{{\rho _3} + {\rho _4} - {\rho _5}}}{{\left( {{\beta _{j,k}} - \frac{\rho }{{{{\bar \gamma }_0}}}} \right)}^{ - {\rho _2}}}\left\{ {\exp \left( {\frac{\rho }{{{{\bar \gamma }_0}}}} \right)\mathbb {EI}\left( { - \frac{\rho }{{{{\bar \gamma }_0}}}} \right) - \exp \left( {\rho {\lambda _2}} \right)\mathbb {EI}\left( { - \rho {\lambda _2}} \right)} \right\},{\rho _3} + {\rho _4} \ne 0}
\end{array}} \right.
\end{equation}
where $\rho=\frac{\rho_5}{\rho_3 + \rho_4}$.}
\end{theorem}
\begin{proof}
The proof is similar with that for Theorem 3 given by Appendix D. Specifically, substituting (11) into (29) and using [28, 3.352.4], one can express $R_{{T_1}}^{{\rm{App}}}$ as in the form of (30). To obtain expression of ${\rm M}^{\cal R}\left( {{\rho _1},{\rho _2},{\rho _3},{\rho _4},{\rho _5}} \right)$ and ${\Lambda ^{\cal R}}\left( {{\rho _1},{\rho _2}} \right)$, one can first apply partial fraction on the fractional terms in the integrals and then use [28, 3.352.4] and/or [28, 3.382.4] on the resultant integrals.
\end{proof}

Again, the expression of ergodic sum rate can be simply derived by exploiting the symmetry between two terminals.

\section{Parameters Optimization Based on the Asymptotic Analysis}
In this section, we optimize the system parameters based on the asymptotic analysis developed in the previous sections.
The optimization problems are constructed which seek to optimally allocate the power at the relay and find optimal relay location, in order to minimize the protocol outage probability. Note that the similar optimization problems can be constructed based on minimizing the sum BER and the results will be identical, since the sum-BER is a linear function of the protocol outage probability in the high SNR regime according to Theorem 5.

\subsection{Power Allocation at the Relay with Fixed Relay Location}
In this subsection, we derive the optimal power allocation at the relay that minimizes the protocol outage probability, where the relay location is fixed. To facilitate the analysis, we let ${\omega _2} = \omega$, then we have ${\omega _1} = 1 - \omega$. The optimization problem can be written as
\begin{equation}
\begin{aligned}
{\omega ^{{\rm{opt}}}} &= \arg \mathop {\min }\limits_\omega  P_{\rm{pro}}^{{\cal O},\rm{Asy}}\left( {{\gamma _{th}}} \right) \\
&{} \mathop  =  \arg \mathop {\min }\limits_\omega \underbrace {\frac{{{B_2} - {B_1}}}{\bar\gamma_2} + \frac{{{B_1} + {C_1}}}{{{\omega }\bar\gamma_2}} + \frac{{{B_1} - {B_2}}}{{\bar\gamma_1}} + \frac{{{B_2} + {C_2}}}{{\left( {1 - {\omega }} \right)\bar\gamma_1}}}_{{{\cal L}}\left( \omega,D \right)} \\
&{} s.t. \hspace{1.5em} 0 \le \omega  \le 1 \\
\end{aligned}
\end{equation}
where $B_i$ and $C_i$ are expressed as ${B_i} = {{\Gamma ''}_{{T_i}}} + 2{{\Gamma '}_{{T_i}}} + 1$ and ${C_i} = \left( {{{\Gamma '}_R} + 1} \right)\left( {{{\Gamma '}_{{T_i}}} + 1} \right)$, respectively. Since the second derivative of ${{{\cal L}}\left( \omega,D \right)} $ with respect to $\omega$ can be expressed as
\begin{equation}
\frac{{{d^2}{{{\cal L}}\left( \omega,D \right)}}}{{d{\omega ^2}}} = \frac{{2\left( {{B_1} + {C_1}} \right)}}{{{\omega ^3}{{\bar \gamma }_2}}} + \frac{{2\left( {{B_2} + {C_2}} \right)}}{{{{\left( {1 - \omega } \right)}^3}{{\bar \gamma }_1}}} > 0
\end{equation}
when $\omega\in [0,1]$, the optimization problem (33) is convex. The optimal power allocation ${\omega ^{{\rm{opt}}}}$ can be obtained by differentiating (33) with respect to $\omega$ and setting the derivative equal to zero, which can be expressed as
\begin{equation}
{\omega ^{{\rm{opt}}}} = \frac{{\sqrt {\left( {{B_1} + {C_1}} \right){{\bar \gamma }_1}} }}{{\sqrt {\left( {{B_1} + {C_1}} \right){{\bar \gamma }_1}}  + \sqrt {\left( {{B_2} + {C_2}} \right){{\bar \gamma }_2}} }}
\end{equation}

From (35), it is seen that when interference power is very small and the noise power is dominant, i.e., ${{\Gamma '}_N} \ll 1$ and ${{\Gamma ''}_N} \ll 1$ for $N=T_1,T_2,R$, we have $B_i \approx 1$ and $C_i \approx 1$. The optimal power allocation reduces to
\begin{equation}
{\omega ^{{\rm{opt}}}}  \approx \frac{{\sqrt {{{\bar \gamma }_1}} }}{{\sqrt {{{\bar \gamma }_1}}  + \sqrt {{{\bar \gamma }_2}} }} = \frac{{\sqrt {{\Omega _1}} }}{{\sqrt {{\Omega _1}}  + \sqrt {{\Omega _2}} }}
\end{equation}
In this case, ${\omega ^{{\rm{opt}}}}$ relies only on the variances of channels between the terminals and relay.
Note that the result is the same with that for 3P-TWR without interference derived in \cite{IEEEhowto:Louie2010}.

When the interference power is large, to understand the effect of interference, we turn to the special case of one interferer at each node, i.e., $L_{T_1}= L_{T_2} = L_R = 1$. In this case, $B_i$ and $C_i$ can be expressed as
\begin{equation}
\begin{aligned}
&   {B_i} = 2{\left( {{P_{I,{T_i}}}{\Omega _{{T_i},1}}} \right)^2} + 2{P_{I,{T_i}}}{\Omega _{{T_i},1}} + 1 = 2{\left( {{P_{I,{T_i}}}d_{{T_i},1}^{ - v}} \right)^2} + 2{P_{I,{T_i}}}d_{{T_i},1}^{ - v} + 1\\
&{} {C_i} = \left( {{P_{I,R}}{\Omega _{R,1}} + 1} \right)\left( {{P_{I,{T_i}}}{\Omega _{{T_i},1}} + 1} \right) = \left( {{P_{I,R}}d_{R,1}^{ - v} + 1} \right)\left( {{P_{I,{T_i}}}d_{{T_i},1}^{ - v} + 1} \right)
\end{aligned}
\end{equation}
where $d_{N,1}$ denotes the distance between node $N$ and the interferer. From (35) and (37), it is seen that the optimal power allocation reduces to (36) when the average received interference powers at two terminals are symmetric, i.e., ${P_{I,{T_1}}}{\Omega _{{T_1},1}}={P_{I,{T_2}}}{\Omega _{{T_2},1}}$. Furthermore, the optimal power allocation number increases as the $P_{I,T_1}$ increasing or the distance between $T_1$ and the interferer decreasing, which indicates that the relay should increase the power used in forwarding the signal from $R$ to terminal $T_1$. Similar result can be found for terminal $T_2$.

\subsection{Relay Location Optimization with Fixed Power Allocation at the Relay}
To minimize the effect of path loss, the relay should be placed on the straight line between $T_1$ and $T_2$. Therefore, we set the distances between $T_i$ ($i=1,2$) and $R$ as $d_{T_1,R}=1-D$ and $d_{T_2,R}=D$, where $D\in(0,1)$. The optimal relay location that minimizes the protocol outage probability with fixed $\omega$ can be derived by solving the following optimization problem
\begin{equation}
\begin{aligned}
{D^{{\rm{opt}}}} & = \arg \mathop {\min }\limits_D {{{\cal L}}\left( \omega,D \right)}\\
&{} s.t. \hspace{1.5em} 0 < D  < 1
\end{aligned}
\end{equation}
It is easy to show that the second derivative of $ {{{\cal L}}\left( \omega,D \right)}$ with respect to $D$ is strictly positive when $D\in(0,1)$. Therefore, the optimal relay location can be derived by differentiating $ {{{\cal L}}\left( \omega,D \right)}$ with respect to $D$ and setting the derivative equal to zero, which can be expressed as
\begin{equation}
{D^{{\rm{opt}}}} = \frac{1}{{{{\left( {\frac{{\omega \left( {1 - \omega } \right)\left( {{B_2} - {B_1}} \right) + \left( {1 - \omega } \right)\left( {{B_1} + {C_1}} \right)}}{{\omega \left( {1 - \omega } \right)\left( {{B_1} - {B_2}} \right) + \omega \left( {{B_2} + {C_2}} \right)}}} \right)}^{\frac{1}{{v - 1}}}} + 1}}
\end{equation}

For the case that the noise power is dominant, we have
\begin{equation}
{D^{{\rm{opt}}}} \approx \frac{{{\omega ^{\frac{1}{{v - 1}}}}}}{{{{\left( {1 - \omega } \right)}^{\frac{1}{{v - 1}}}} + {\omega ^{\frac{1}{{v - 1}}}}}}
\end{equation}
Note that the value of path-loss exponent $v$ is normally in the range from 2 to 6 \cite{IEEEhowto:Rappaport}. As a result, we can conclude that the relay should be placed near $T_1$ ($D>0.5$) when more relay power is allocated to forward the signal from relay to $T_1$ ($\omega > 0.5$) when the noise power is dominant. Similarly, we have $D<0.5$ when $\omega < 0.5$.

When the interference is large, we focus on the special case of one interferer at each node. In this case, $B_i$ and $C_i$ can be expressed as in (37). From (39), when the average received interference powers at $T_1$ and $T_2$ are symmetric, i.e., ${P_{I,{T_1}}}{\Omega _{{T_1},1}}={P_{I,{T_2}}}{\Omega _{{T_2},1}}$, we can see that the optimal relay location reduces to (40) after a few manipulations. This indicates that the optimal $D$ is decided by the power allocation at the relay in this case. When the average received interference powers at $T_1$ and $T_2$ are asymmetric, i.e., ${P_{I,{T_1}}}{\Omega _{{T_1},1}}={P_{I,{T_2}}}{\Omega _{{T_2},1}}$, we consider the case $\omega=0.5$. The optimal relay location reduces to
\begin{equation}
{D^{{\rm{opt}}}} = \frac{{{{\left( {{B_1} + {B_2} + 2{C_2}} \right)}^{\frac{1}{{v - 1}}}}}}{{{{\left( {{B_1} + {B_2} + 2{C_1}} \right)}^{\frac{1}{{v - 1}}}} + {{\left( {{B_1} + {B_2} + 2{C_2}} \right)}^{\frac{1}{{v - 1}}}}}}
\end{equation}
According to the expressions of $B_i$ and $C_i$, the relay should be placed near the terminal with larger average received interference power in order to minimize the protocol outage (or sum BER) performance.

\subsection{Joint Optimization of Power Allocation at the Relay and Relay Location}
As shown in (35) and (39), the optimal $\omega$ and $D$ are not independent. As a result, jointly optimizing these two parameters can achieve better performance. The optimization problem can be formulated as
\begin{equation}
\begin{aligned}
&   \left( {{\omega ^{{\rm{opt}}}},{D^{{\rm{opt}}}}} \right)  =\arg \mathop {\min }\limits_{\omega,D} {{{\cal L}}\left( \omega,D \right)} \\
&{} s.t. \hspace{1.5em} 0 \le \omega  \le 1, 0 < D  < 1
\end{aligned}
\end{equation}

\newtheorem{Corollary}{\emph{Corollary}}
\begin{Corollary}\label{t1}
\emph{When ${{\Gamma '}_{{T_1}}}={{\Gamma '}_{{T_2}}}$ and ${{\Gamma ''}_{{T_1}}}={{\Gamma ''}_{{T_2}}}$, the optimal power allocation at the relay and optimal relay location are $\omega=0.5$ and $D=0.5$.}
\end{Corollary}

\begin{proof}
When ${{\Gamma '}_{{T_1}}}={{\Gamma '}_{{T_2}}}$ and ${{\Gamma ''}_{{T_1}}}={{\Gamma ''}_{{T_2}}}$, we have $B_1-B_2=0$ and $C_1-C_2=0$. In this case, differentiating ${{{\cal L}}\left( \omega,D \right)}$ with respect $\omega$ and $D$ twice, one can show that the Hessian matrix of ${{{\cal L}}\left( \omega,D \right)}$ is positive semi-definite when $v\ge2$. Therefore, solving equations $\frac{{\partial {\cal L}\left( {\omega ,D} \right)}}{{\partial \omega }} = 0$ and $\frac{{\partial {\cal L}\left( {\omega ,D} \right)}}{{\partial D}} = 0$ jointly, one can obtain the result in Corollary 1.
\end{proof}

Unfortunately, for the general case of ${{\Gamma '}_{{T_1}}} \ne {\Gamma '}_{{T_2}}$ or ${{\Gamma ''}_{{T_1}}} \ne {\Gamma ''}_{{T_2}}$, we can not prove the optimization problem is convex. Even it can be proved, the solution is hard to be obtained with closed-form in this case. As a result, with (35) and (39) in the previous subsections, we resort to a simple but still efficient alternating optimization approach \cite{IEEEhowto:Pun2004GC},\cite{IEEEhowto:Ziskind1988} to deal with this problem. The algorithm is given in the below
\begin{enumerate}
\item Initialize $D$ as $D = D^{(0)}=0.5$.
\item At the $l$th iteration ($l \ge 1$), update $\omega=\omega^{(l)}$ using (35), where $D$ is set to $D=D^{(l-1)}$.
\item Update $D=D^{(l)}$ using (39), where $\omega$ is set to $\omega=\omega^{(l)}$.
\item Set $l=l+1$ and go back to step 2), until the algorithm reaches the preassigned number of {iterations\footnotemark[7]}.
\end{enumerate}
{\footnotetext[7]{Strictly speaking, the algorithm should be terminated when ${{{\cal L}}\left( \omega,D \right)}$ does not change significantly. However, to avoid computing ${{{\cal L}}\left( \omega,D \right)}$ at each iteration, we fix the number of iterations.}}

Since some minimizations are performed at each iteration, the value of ${{{\cal L}}\left( \omega,D \right)}$ can not increase. As a result, the algorithm is bound to converge to a local minimum \cite{IEEEhowto:Pun2004GC},\cite{IEEEhowto:Ziskind1988}. As will be shown by the simulations, with only a few iterations, the proposed algorithm can achieve almost the same performance compared with the scheme using optimal $\omega$ and $D$ obtained by exhaustion method.

\section{Simulation Results and Discussion}

In this section, we present the simulation results to verify our theoretical analyses in the previous sections. We assume that the terminals and relay are placed in a straight line and the relay is set between $T_1$ and $T_2$. The normalized distance between $T_1$ and $T_2$ is set to one. Moreover, the variance of $c_{N,k}$ ($N\in\{T_1,T_2,R\}$) is assumed to be evenly distributed on the interval $[0.1,1]$, then we have $\Omega_{N,k}=0.1+\frac{0.9}{L_{N}-1}(k-1)$ for $L_N\ge2$. For $L_N=1$, we set $\Omega_{N,1}= 1$.

In Fig. 3-Fig. 5, the performance is simulated where the interference at two terminal is symmetric, i.e., $P_{I,T_1}=P_{I,T_2}$ and $L_{T_1}=L_{T_2}$. We present only the performance when $\omega=0.5$ and $D=0.5$, since this setup leads to the optimal protocol outage and sum BER performances in this case according to Corollary 1.

\begin{figure}[t]
\centering
\includegraphics[width=8.0cm]{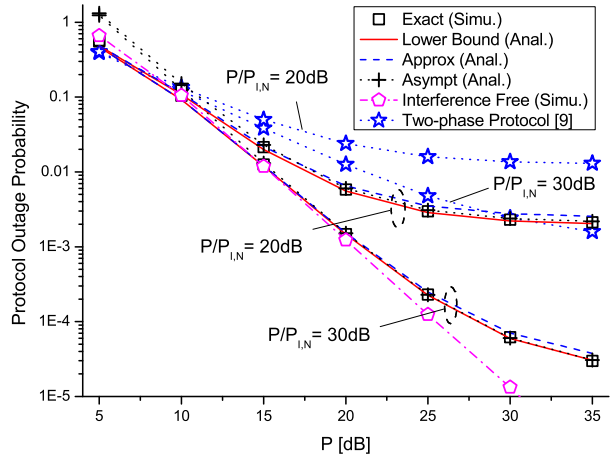}
\caption{Protocol outage performance versus $P$, $L_{N}=2$, $N\in\{T_1,T_2,R\}$, $\gamma_{th}=7$ (Corresponding to 1bit/s/Hz target rate), $\omega = 0.5$, $D= 0.5$.} \label{fig:graph}
\end{figure}

\begin{figure}[t]
\centering
\includegraphics[width=8.0cm]{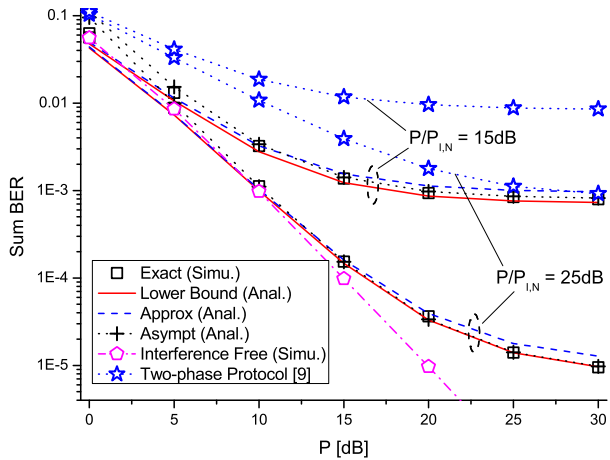}
\caption{Sum BER performance versus $P$, $L_{N}=5$, $N\in\{T_1,T_2,R\}$, $\omega = 0.5$, $D= 0.5$.} \label{fig:graph}
\end{figure}

In Fig. 3 and Fig. 4, the protocol outage and sum BER performances of 3P-TWR protocol are plotted as a function of the transmit power $P$. The performances of the interference free scenario are also presented as benchmark. We can see that the proposed lower bounds yield results in good agreement with the exact results derived by Monte Carlo simulations in the whole observation interval. Meanwhile, the approximate expressions perform better than the asymptotic expressions in the low SNR region whereas the asymptotic expressions do better in the moderate and high SNR regions. Due to this observation, one can estimate the protocol outage and sum BER performances efficiently by selectively using the approximate and asymptotic expressions depending on the transmitted SNR. From the figures, performance floors can be observed in the high SNR region which indicates that the achievable diversity of 3P-TWR protocol in the interference-limited scenario is zero. Moreover, as shown in the figure, the 3P-TWR protocol performs better in protocol outage and sum-BER performances compared with the 2P-TWR protocol in the interference-limited scenario. The result suggests that the three-phase protocol is a good choice when network has a strict requirement on reliability.

\begin{figure}[t]
\centering
\includegraphics[width=8.0cm]{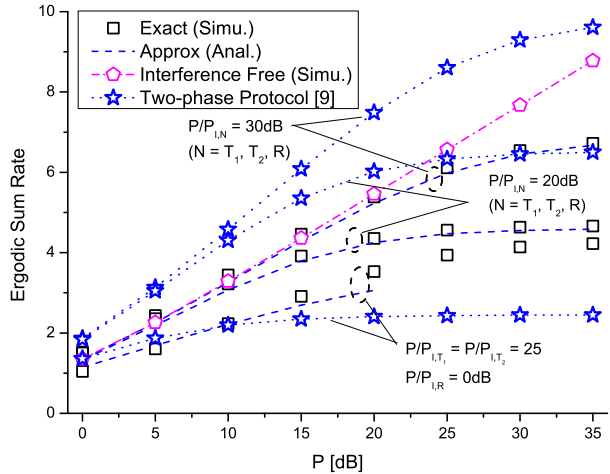}
\caption{Ergodic sum rate performance versus $P$, $L_{N}=5$, $N\in\{T_1,T_2,R\}$, $\omega _1= 0.5$, $D_1=0.5$.} \label{fig:graph}
\end{figure}

Fig. 5 depicts the ergodic sum rate performance of 3P-TWR protocol against the transmitted power $P$. As shown in the figure, the ergodic sum rate degrades as the interference power increasing as expected. The performance floor in the high SNR region is because we assume the ratio of useful power to interference power is constant. Moreover, it is interesting to see that when the interference power at the relay is much larger than that at the terminals, the 3P-TWR protocol outperforms 2P-TWR protocol in ergodic sum rate, which is in sharp contrast with the situation in interference-free scenario. This is because the 2P-TWR protocol uses only the terminal-relay-terminal channel, whose received SINR is degraded greatly by the interference at the relay. However, the 3P-TWR protocol exploits the direct channel, thus can achieve better performance.

The performances of 3P-TWR protocol with power and relay location optimization are testified in Fig. 6-Fig. 8. We set $P/P_{I,R} = P/P_{I,T_1}$ and $P/P_{I,T_1} \ne P/P_{I,T_2}$, which represents an asymmetric interference power profile at two terminals.

Fig. 6 shows the convergence property of the proposed joint optimization scheme in Sec. VII-C. From the figure, the performance of joint optimization converges to the optimal performance obtained by exhaustion method with three iterations. Moreover, it is seen that the scheme which optimizes only the relay location achieves nearly the same performance with the scheme without optimization (i.e., with zero iteration) when $\omega=0.5$ and the interference power at the relay is small. This is because, in this case, we have $B_1+B_2>2C_1$ and $B_1+B_2>2C_2$. As a result, the optimal relay location will be very close to 0.5 according to (41).

Fig. 7 and Fig. 8 present the protocol outage and sum BER performances as a function of $P$. From the figures, we can see that the joint optimization with a few iterations can provide significant performance gain compared with the scheme without optimization. From the perspective of the practical implementation, it is seen that the optimal number of iterations is two in this case, since the performance gain provided by the third iteration is quite small and negligible.

\begin{figure}[t]
\centering
\includegraphics[width=8.0cm]{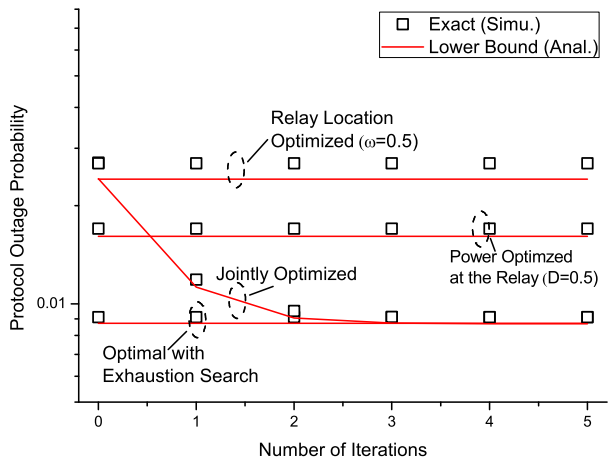}
\caption{Protocol outage performance versus number of iterations, $L_{N}=5$, $N\in\{T_1,T_2,R\}$, $P/P_{I,R} = P/P_{I,T_1}=25$dB, $P/P_{I,T_2}=15$dB, $\gamma_{th}=7$.} \label{fig:graph}
\end{figure}

\begin{figure}[t]
\centering
\includegraphics[width=8.0cm]{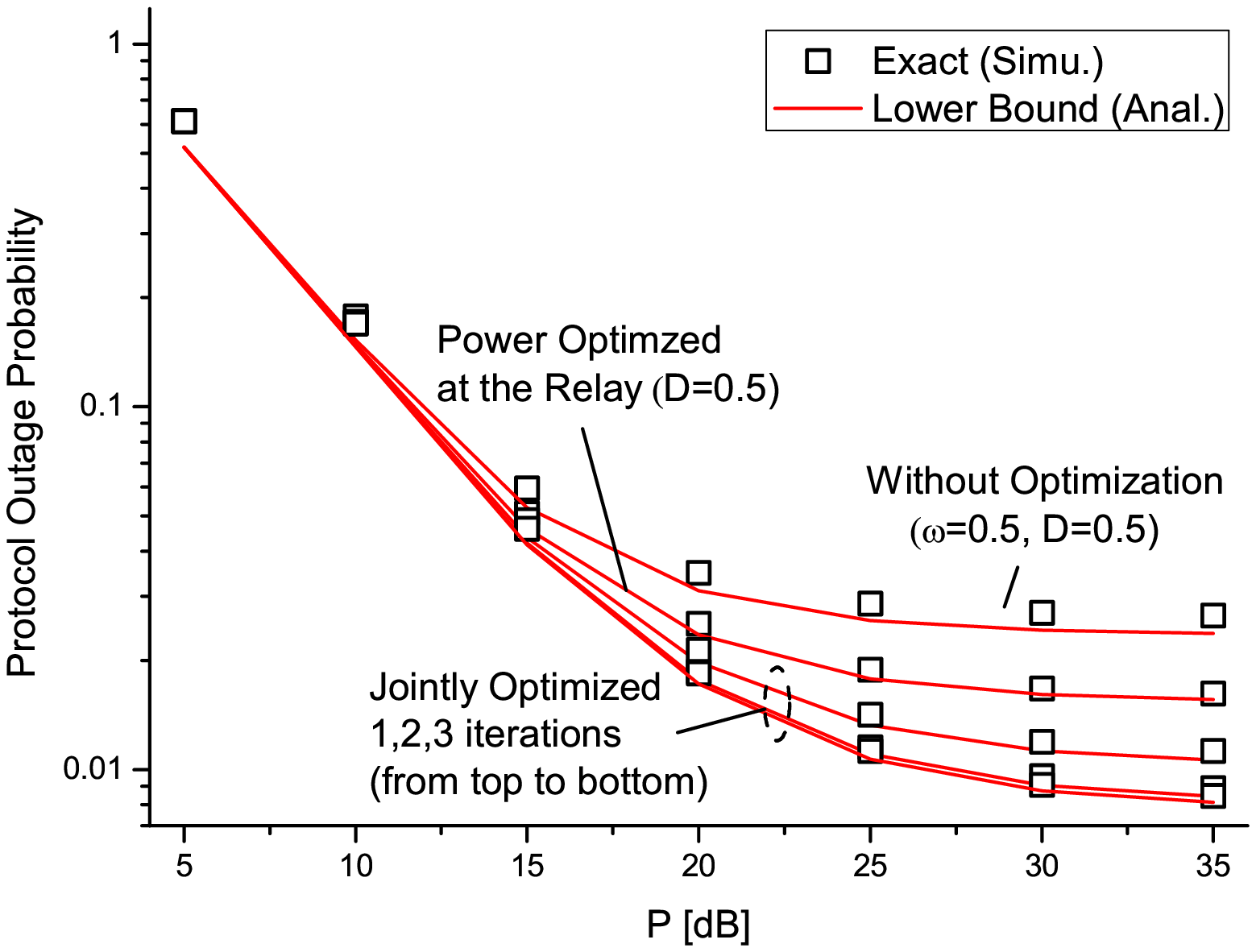}
\caption{Protocol outage performance versus $P$, $L_{N}=5$, $N\in\{T_1,T_2,R\}$, $P/P_{I,R} = P/P_{I,T_1}=25$dB, $P/P_{I,T_2}=15$dB, $\gamma_{th}=7$.} \label{fig:graph}
\end{figure}

\begin{figure}[t]
\centering
\includegraphics[width=8.0cm]{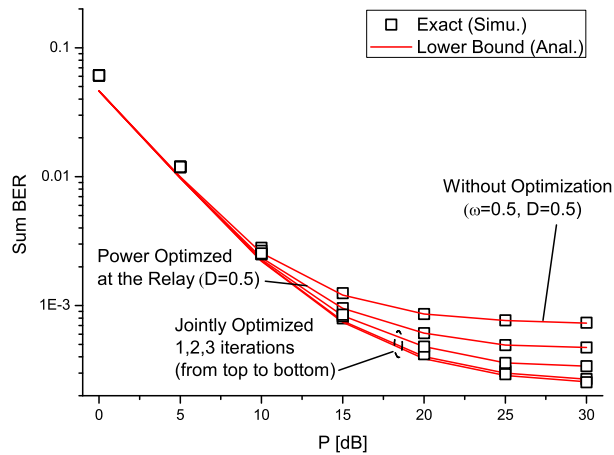}
\caption{Sum BER performance versus $P$, $L_{N}=5$, $N\in\{T_1,T_2,R\}$, $P/P_{I,R} = P/P_{I,T_1}=20$dB, $P/P_{I,T_2}=10$dB, $N\in\{T_1,T_2,R\}$.} \label{fig:graph}
\end{figure}

The effects of imperfect CSI are analyzed in Fig. 9 and Fig. 10. The actual and the estimated channels are modeled as \cite{IEEEhowto:Yang2013TC} ${h_i} = {\hat h_i} + {e_i}$ ($i\in\{1,2,3\}$) and ${c_{N,k}} = {\hat c_{N,k}} + {e_{N,k}}$ ($N\in\{T_1,T_2,N\}$), where ${\hat h_i}$ and ${\hat c_{N,k}}$ denote the estimates of ${h_i}$ and ${c_{N,k}}$, respectively. ${e_i}$ (${e_{N,k}}$) denotes the estimation error of ${h_i}$ (${c_{N,k}}$), which is an independent zero-mean complex Gaussian RV with variance $\sigma_h \Omega_i$ ($\sigma_{c}\Omega_{N,k}$) \cite{IEEEhowto:Yang2013TC}. The gain at the relay and received SINRs at the terminals can be computed using the method in \cite{IEEEhowto:Yang2013TC}. Since the results for protocol outage and sum BER are similar, we present only the performances for protocol outage and ergodic sum rate in the following.

\begin{figure}[t]
\centering
\includegraphics[width=8.0cm]{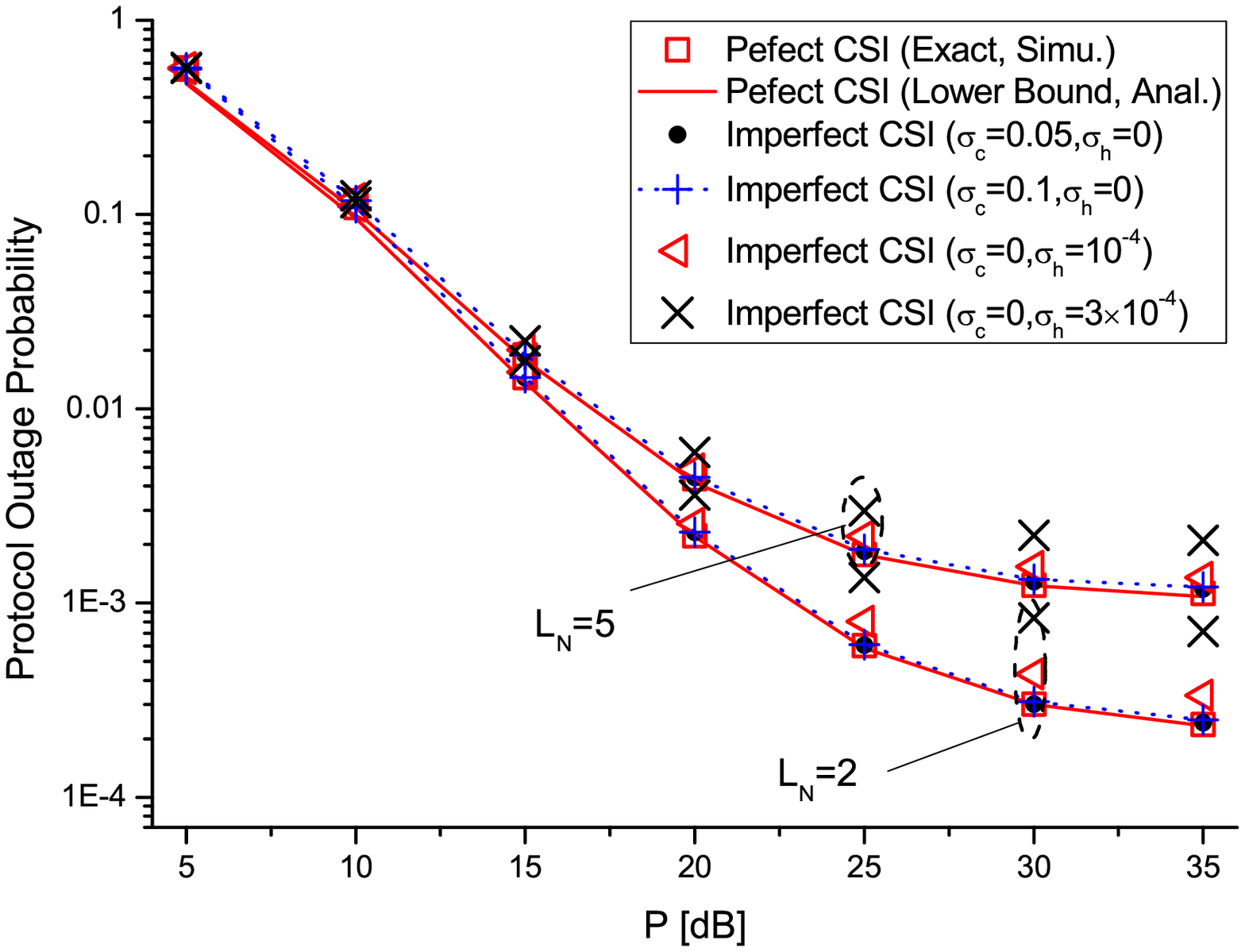}
\caption{Comparison of protocol outage performances with perfect and imperfect CSI, $P/P_{I,N} = 25$dB, $N\in\{T_1,T_2,R\}$, $\gamma_{th}=7$, $\omega=0.5$, $D=0.5$.} \label{fig:graph}
\end{figure}

\begin{figure}[t]
\centering
\includegraphics[width=8.0cm]{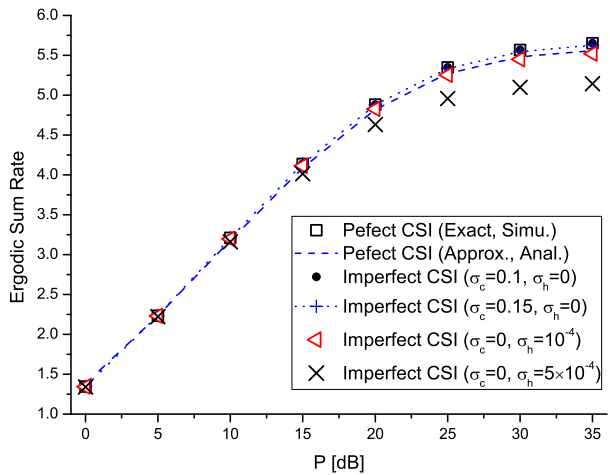}
\caption{Comparison of ergodic sum rate performances with perfect and imperfect CSI, $L_{N}=5$, $N\in\{T_1,T_2,R\}$, $P/P_{I,N} = 25$dB, $\omega = 0.5$, $D=0.5$.} \label{fig:graph}
\end{figure}

From Fig. 9 and Fig. 10, we can see that the performance with perfect CSI provides a bound for other practical scenarios with imperfect channel estimation. Meanwhile, it is interesting to see that an estimation error on the channel coefficient of interference channel (i.e., $c_{N,k}$) will not affect the protocol outage and ergodic sum rate performances significantly. This is because, in this case, a small fraction of the interference signal is actually treated as noise in the calculation of received SINRs at the terminals \cite{IEEEhowto:Yang2013TC}, which will not change the received SINR greatly. However, it is seen that a very small estimation error on the channel coefficients between terminals and relay (i.e., $h_i$, $i\in\{0,1,2\}$) will degrade the performance considerably, since the estimation error on $h_i$ introduces new interference.

\section{Conclusion}

The effect of CCI on the three-phase AF-based TWR protocol is considered in this paper for Rayleigh fading channels. The lower bounds, approximate expressions and asymptotic expressions for protocol outage probability and sum BER are derived. Moreover, the approximate expression for ergodic sum rate is derived. These expressions are valid for arbitrary positive numbers of interferers at the terminals and relay. The performances of 2P-TWR protocol and 3P-TWR protocol are compared. The results show that when the interference power at the relay is much larger than that at the terminals, the 3P-TWR protocol outperforms 2P-TWR protocol in ergodic sum rate, which is in sharp contrast with the situation in interference-free scenario. The system parameters, i.e., the power allocation at the relay and relay location, are optimized based on the asymptotic expressions in order to minimize the protocol outage and sum BER performances in interference-limited scenario. The results show that when the average received interference powers at two terminals are asymmetric, jointly optimizing the power and relay location can give rise to the best performance.

\section*{Appendix A: Proof of (5)}
We prove the result for $T_1$ and the result for $T_2$ is similar. According to the principle of MRC, the combined signal at $T_1$ can be expressed as
\begin{equation*}
y_{{T_1}}^{MRC} = {{\cal C}_1}y_{{T_1}}^{[3]} + {{\cal C}_2}y_{{T_1}}^{[2]}
\end{equation*}
where the combining coefficients can be expressed as
\begin{equation*}
\begin{aligned}
&   {{\cal C}_1} = \frac{{P{{\cal A}_2}h_1^*h_2^*}}{{P{\gamma _1}\left( {{\cal A}_1^2 + {\cal A}_2^2} \right)\left( {{\Gamma _R} + 1} \right) + {\Gamma _{{T_1}}} + 1}}\\
&{} {{\cal C}_2} = \frac{{\sqrt P h_0^*}}{{{\Gamma _{{T_1}}} + 1}}
\end{aligned}
\end{equation*}
According to the expressions of ${{\cal C}_1}$ and ${{\cal C}_2}$, $y_{{T_1}}^{MRC}$ can be rewritten as
\begin{equation*}
\begin{aligned}
y_{{T_1}}^{MRC} &= \left\{ {\frac{{{P^2}{\cal A}_2^2{\gamma _1}{\gamma _2}}}{{P{\gamma _1}\left( {{\cal A}_1^2 + {\cal A}_2^2} \right)\left( {{\Gamma _R} + 1} \right) + {\Gamma _{{T_1}}} + 1}} + \frac{{P{\gamma _0}}}{{{\Gamma _{{T_1}}} + 1}}} \right\}{S_2}\\
&{} + \frac{{P{{\cal A}_2}h_1^*h_2^*}}{{P{\gamma _1}\left( {{\cal A}_1^2 + {\cal A}_2^2} \right)\left( {{\Gamma _R} + 1} \right) + {\Gamma _{{T_1}}} + 1}}\\
&{} \times \left( {\sqrt P {h_1}{{\cal A}_1}\sum\limits_{k = 1}^{{L_R}} {\sqrt {{P_{I,R}}} {c_{R,k}}I_{R,k}^{[1]}}  + \sqrt P {h_1}{{\cal A}_2}\sum\limits_{k = 1}^{{L_R}} {\sqrt {{P_{I,R}}} {c_{R,k}}I_{R,k}^{[2]}}  + \sum\limits_{k = 1}^{{L_{{T_1}}}} {\sqrt {{P_{I,{T_1}}}} {c_{{T_1},k}}} I_{{T_1},k}^{[3]}} \right)\\
&{} + \frac{{\sqrt P h_0^*}}{{{\Gamma _{{T_1}}} + 1}}\left( {\sum\limits_{k = 1}^{{L_{{T_1}}}} {\sqrt {{P_{I,{T_1}}}} {c_{{T_1},k}}I_{{T_1},k}^{[2]}}  + n_{{T_1}}^{[2]}} \right)
\end{aligned}
\end{equation*}
As a result, the received SINR at $T_i$ can be expressed as
\begin{equation*}
{\Upsilon _{{T_1}}} = \frac{{{\gamma _0}}}{{{\Gamma _{{T_1}}} + 1}} + \frac{{{\cal A}_2^2{\gamma _1}{\gamma _2}}}{{\left( {{\cal A}_1^2 + {\cal A}_2^2} \right)\left( {{\Gamma _R} + 1} \right){\gamma _1} + {\Gamma _{{T_1}}} + 1}}
\end{equation*}
Substituting the expressions ${\cal A}_1$ and ${\cal A}_2$ into the above equation, we have
\begin{equation*}
\begin{aligned}
& {\Upsilon _{{T_1}}} = \frac{{{\gamma _0}}}{{{\Gamma _{{T_1}}} + 1}} + \frac{{{\omega _2}{\gamma _1}{\gamma _2}}}{{\left( {{\Gamma _R} + 1} \right){\gamma _1} + \left( {{\omega _1}{\gamma _1} + {\omega _2}{\gamma _2} + {\Gamma _R} + 1} \right)\left( {{\Gamma _{{T_1}}} + 1} \right)}}\\
&{} = \frac{{{\gamma _0}}}{{{\Gamma _{{T_1}}} + 1}} + \frac{{\frac{{{\gamma _1}}}{{{\Gamma _{{T_1}}} + 1}}\frac{{{\omega _2}{\gamma _2}}}{{{\Gamma _R} + {\omega _1}{\Gamma _{{T_1}}} + {\omega _1} + 1}}}}{{\frac{{{\gamma _1}}}{{{\Gamma _{{T_1}}} + 1}} + \frac{{{\omega _2}{\gamma _2}}}{{{\Gamma _R} + {\omega _1}{\Gamma _{{T_1}}} + {\omega _1} + 1}} + \frac{{{\Gamma _R} + 1}}{{{\Gamma _R} + {\omega _1}{\Gamma _{{T_1}}} + {\omega _1} + 1}}}}\\
&{} = {\Upsilon _{{T_1},D}} + \frac{{{\Upsilon _{{T_1},1}}{\Upsilon _{{T_1},2}}}}{{{\Upsilon _{{T_1},1}} + {\Upsilon _{{T_1},2}} + \frac{{{\Gamma _R} + 1}}{{{\Gamma _R} + {\omega _1}{\Gamma _{{T_1}}} + {\omega _1} + 1}}}}\\
&{} \simeq {\Upsilon _{{T_1},D}} + \frac{{{\Upsilon _{{T_1},1}}{\Upsilon _{{T_1},2}}}}{{{\Upsilon _{{T_1},1}} + {\Upsilon _{{T_1},2}}}}
\end{aligned}
\end{equation*}
Note that in the derivation, we have assumed that the interference signal at the relay in different phases, i.e., $I_{R,k}^{[1]}$ and $I_{R,k}^{[2]}$ are independent. This is reasonable because we consider the cases that the time duration of each phase is much longer than one codewords.

\section*{Appendix B: Proof of Lemma 1}
Since $\gamma_0$ is an exponential RV with mean $\bar \gamma_0$, it is easy to verify that $\Pr \left( {\left. {{\Upsilon _{{T_1},D}} > \gamma } \right|{\Gamma _{{T_1}}}} \right)$ can be expressed as
\begin{equation}
\Pr \left( {\left. {{\Upsilon _{{T_1},D}} > \gamma } \right|{\Gamma _{{T_1}}}} \right) = \Pr \left( {\left. {{\gamma _0} > \gamma \left( {{\Gamma _{{T_1}}} + 1} \right)} \right|{\Gamma _{{T_1}}}} \right) = \exp \left( { - \frac{\gamma }{{{\gamma _0}}}\left( {{\Gamma _{{T_1}}} + 1} \right)} \right)
\end{equation}
Inserting (10) and (43) into (9) and solving the resultant integral, one can obtain
\begin{equation}
{{\cal F}_1}\left( \gamma  \right) = \exp \left( { - \frac{\gamma }{{{{\bar \gamma }_0}}}} \right)\sum\limits_{j } {{\phi _{T_1,j}}} \frac{{{{\bar \gamma }_0}}}{{\gamma  + {{{{\bar \gamma }_0}} \mathord{\left/
{\vphantom {{{{\bar \gamma }_0}} {{\xi_{{T_1},j}}}}} \right.
\kern-\nulldelimiterspace} {{\xi_{{T_1},j}}}}}}
\end{equation}
Moreover, the conditional probability $\Pr \left( {\left. {{\Upsilon _{{T_1},D}} < \gamma ,{\Upsilon _{T_i}^m} > \gamma  - {\Upsilon _{{T_1},D}}} \right|{\Gamma _{{T_1}}},{\Gamma _{{T_R}}}} \right)$ can be rewritten as
\begin{equation}
\begin{aligned}
& \Pr   \left( {\left. {{\Upsilon _{{T_1},D}} < \gamma ,{\Upsilon _{T_i}^m} > \gamma  - {\Upsilon _{{T_1},D}}} \right|{\Gamma _{{T_1}}},{\Gamma _{{T_R}}}} \right) \\
&{} =\! \int\limits_0^\gamma  {\int\limits_{\gamma  - z}^\infty  {{f_{{\Upsilon _{T_i}^m}\left| {\left\{ {{\Gamma _{{T_1}}},{\Gamma _{{T_R}}}} \right\}} \right.}}\left( x \right)} {f_{\left. {{\Upsilon _{{T_1},D}}} \right|{\Gamma _{{T_1}}}}}\left( z \right)dxdz}  \! =\!  \int\limits_0^\gamma  {\left( {1 - {F_{{\Upsilon _{T_i}^m}\left| {\left\{ {{\Gamma _{{T_1}}},{\Gamma _{{T_R}}}} \right\}} \right.}}\left( {\gamma  - z} \right)} \right)} {f_{\left. {{\Upsilon _{{T_1},D}}} \right|{\Gamma _{{T_1}}}}}\left( z \right)dz
\end{aligned}
\end{equation}
where ${f_{{{\Upsilon _{T_i}^m} \left| {\left\{ {\Gamma_R,\Gamma_{T_1}} \right\}} \right.}}}\left( x  \right)$ is the PDF of $\Upsilon _{m}$ conditioned on $\Gamma_R$ and $\Gamma_{T_1}$. The second step is obtained by solving the integral over $x$. ${F_{{{\Upsilon _{T_i}^m} \left| {\left\{ {\Gamma_R,\Gamma_{T_1}} \right\}} \right.}}}\left( x  \right)$ is the CDF of ${\Upsilon _{T_i}^m}$ conditioned on $\Gamma_R$ and $\Gamma_{T_1}$, which can be written as
\begin{equation}
\begin{aligned}
{F_{{{\Upsilon _{T_i}^m}}\left| {\left\{ {{\Gamma _R},{\Gamma _{{T_1}}}} \right\}} \right.}} & \left( x \right) = 1 - \prod\limits_{i = 1}^2 {\left( {1 - {F_{{\gamma _{{T_1},i}}|\left\{ {{\Gamma _R},{\Gamma _{{T_1}}}} \right\}}}\left( x \right)} \right)} \\
&{} = 1 - \exp \left( { - \frac{{{\Gamma _{{T_1}}} + 1}}{{{{\bar \gamma }_1}}}x} \right)\exp \left( { - \frac{{{\Gamma _R} + {\omega _1}{\Gamma _{{T_1}}} + 1 + {\omega _1}}}{{{\omega _2}{{\bar \gamma }_2}}}x} \right)
\end{aligned}
\end{equation}
and ${f_{\left. {{\Upsilon _{{T_1},D}}} \right|{\Gamma _{{T_1}}}}}\left( z \right)$ is the PDF of $\Upsilon _{T_1,D}$ conditioned on $\Gamma_{T_1}$, which can be expressed as
\begin{equation}
{f_{\left. {{\Upsilon _{{T_1},D}}} \right|{\Gamma _{{T_1}}}}}\left( z \right) = \frac{{{\Gamma _{{T_1}}} + 1}}{{{{\bar \gamma }_0}}}\exp \left( { - \frac{{{\Gamma _{{T_1}}} + 1}}{{{{\bar \gamma }_0}}}z} \right)
\end{equation}
Substituting (10) and (45)-(47) into (9) and interchanging the integration order, we can obtain
\begin{equation}
\begin{aligned}
{{\cal F}_2}\left( \gamma  \right)  &= \frac{1}{{{{\bar \gamma }_0}}}\exp \left( { - \frac{1}{{{{\bar \gamma }_1}}}\gamma  - \frac{{1 + {\omega _1}}}{{{\omega _2}{{\bar \gamma }_2}}}\gamma } \right)  \! \sum\limits_j \!{\sum\limits_k \!{{\phi _{T_1,j}}{\phi _{R,k}}\!\int\limits_0^\gamma  {\!\exp\! \left(\! { - {\Phi _2}z} \!\right)} } }\! \!\int\limits_0^\infty \!\! {\exp\! \left( { - \left[ {\frac{{\gamma \! -\! z}}{{{\omega _2}{{\bar \gamma }_2}}} \!+\! \frac{1}{{{\xi_{R,k}}}}} \right]\!s} \!\right)} ds\\
&{} \!\times \! \int\limits_0^\infty \! {\left( {t + 1} \right)\!\exp \!\left(\! { - \left[ {{\Phi _1}z + \left( {\frac{1}{{{{\bar \gamma }_1}}} \!+\! \frac{{{\omega _1}}}{{{\omega _2}{{\bar \gamma }_2}}}} \right)\gamma  + \frac{1}{{{\xi_{{T_1},j}}}}} \right]t} \!\right)} dtdz
\end{aligned}
\end{equation}
Solving the integrals with respect to $s$ and $t$, we can yield
\begin{equation}
\begin{array}{ll}
{{\cal F}_2}\left( \gamma  \right) =  \frac{{{\omega _2}}}{{{{\bar \gamma }_0}}}\exp \left( { - \frac{1}{{{{\bar \gamma }_1}}}\gamma  - \frac{{1 + {\omega _1}}}{{{\omega _2}{{\bar \gamma }_2}}}\gamma } \right)\sum\limits_j {\sum\limits_k {{\phi _{T_1,j}}{\phi _{R,k}}\frac{{{{\bar \gamma }_0}}}{{\gamma  + {{\bar \gamma }_0}{\beta _{j,k}}}}} } \\
\times \left\{ {\left( {1 + \frac{{{{\bar \gamma }_0}}}{{\gamma  + {{\bar \gamma }_0}{\beta _{j,k}}}}} \right)} \right.\Psi \left( { - \frac{1}{{{{\bar \gamma }_2}}},\frac{1}{{{{\bar \gamma }_2}}},\frac{{{\omega _2}}}{{{\xi_{R,k}}}}} \right) + \left( {\frac{1}{{{\omega _2}}} + \frac{{{\Phi _1}{{\bar \gamma }_0}{{\bar \gamma }_2}}}{{\gamma  + {{\bar \gamma }_0}{\beta _{j,k}}}}} \right)\Psi \left( {{\Phi _1},\frac{1}{{{{\bar \gamma }_1}}} + \frac{{{\omega _1}}}{{{\omega _2}{{\bar \gamma }_2}}},\frac{1}{{{\xi_{{T_1},j}}}}} \right)\\
\left. { + \left( {\frac{{{{\bar \gamma }_2}}}{{\left( {\frac{1}{{{{\bar \gamma }_1}}} + \frac{{{\omega _1}}}{{{\omega _2}{{\bar \gamma }_2}}}} \right)\gamma  + \frac{1}{{{\xi_{{T_1},j}}}}}} - \frac{{{{\bar \gamma }_2}}}{{\frac{1}{{{{\bar \gamma }_0}}}\gamma  + \frac{1}{{{\xi_{{T_1},j}}}}}}\exp \left( { - {\Phi _2}\gamma } \right)} \right)} \right\}
\end{array}
\end{equation}
where $\Psi \left( {{\rho _1},{\rho _2},{\rho _3}} \right)$ is expressed as $\Psi \left( {{\rho _1},{\rho _2},{\rho _3}} \right) = \int_0^\gamma  {\exp \left( { - {\Phi _2}z} \right)} \frac{1}{{{\rho _1}z + {\rho _2}\gamma  + {\rho _3}}}dz$. Taking a closer look at (49), it is easy to verify that $ {{\cal F}_2}\left( \gamma  \right)$ can be rewritten as
\begin{equation}
\begin{aligned}
& {{\cal F}_2}\left( \gamma  \right) = \frac{{{\omega _2}}}{{{{\bar \gamma }_0}}}\sum\limits_{j} {\sum\limits_{k} {{\phi _{{T_1},j}}{\phi _{R,k}}} } \left( {\rm{M}}\left( {1,1, - \frac{1}{{{{\bar \gamma }_2}}},\frac{1}{{{{\bar \gamma }_2}}},\frac{{{\omega _2}}}{{{\xi_{R,k}}}}} \right) + {\rm{M}}\left( {1,2, - \frac{1}{{{{\bar \gamma }_2}}},\frac{1}{{{{\bar \gamma }_2}}},\frac{{{\omega _2}}}{{{\xi_{R,k}}}}} \right)\right. \\
&{}   \left. + {\rm{M}}\left( {\frac{1}{{{\omega _2}}},1,{\Phi _1},{\lambda _1},\frac{1}{{{\xi_{{T_1},j}}}}} \right) {+ {\rm{M}}\left( {{\Phi _1}{{\bar \gamma }_2},2,{\Phi _1},{\lambda _1},\frac{1}{{{\xi_{{T_1},j}}}}} \right) + \Lambda \left( {{\lambda _1},{\lambda _2}} \right) - \Lambda \left( {\frac{1}{{{{\bar \gamma }_0}}},\frac{1}{{{{\bar \gamma }_0}}}} \right)} \right)
\end{aligned}
\end{equation}
where $\Lambda(\rho_1,\rho_2)$ is given by (13) and ${\rm{M}}\left( {{\rho _1},{\rho _2},{\rho _3},{\rho _4},{\rho _5}} \right)$ is expressed as
\begin{equation}
\begin{aligned}
 {\rm{M}} &= {\rho _1}\exp \left( { - {\lambda _2}\gamma } \right){\left( {\frac{1}{{{{\bar \gamma }_0}}}\gamma  + {\beta _{j,k}}} \right)^{ - {\rho _2}}}\Psi \left( {{\rho _3},{\rho _4},{\rho _5}} \right) \\
&{} = {\rho _1}\exp \left( { - {\lambda _2}\gamma } \right){\left( {\frac{1}{{{{\bar \gamma }_0}}}\gamma  + {\beta _{j,k}}} \right)^{ - {\rho _2}}}\int\limits_0^\gamma  {\frac{{\exp \left( { - {\Phi _2}z} \right)}}{{{\rho _3}z + {\rho _4}\gamma  + {\rho _5}}}} dz
 \end{aligned}
\end{equation}
To solve the integral introduced by $\Psi \left( {{\rho _3},{\rho _4},{\rho _5}} \right)$, we apply Taylor series expansion $\frac{1}{{{\rho _3}z + {\rho _4}\gamma  + {\rho _5}}} = \frac{1}{{{\rho _3}}}\sum\limits_{l = 0}^\infty  {{{\left( { - 1} \right)}^l}{{\left( {\frac{{{\rho _3}}}{{{\rho _4}\gamma  + {\rho _5}}}} \right)}^{l + 1}}{z^l}} $. Then based on [28, 3.381.1], (51) can be solved as in (12). Finally, substituting (44) and (50) into (9), we can obtain the result in Lemma 1.

\section*{Appendix C: Proof of Lemma 3}
According to \cite{IEEEhowto:Wang}, the asymptotic expression can be derived by performing McLaurin series expansion on $F_{\Upsilon _{{T_1}}^{\rm{UB}}}\left( \gamma  \right)$ and taking only the first two order terms. Here the major difficulty is due to the series expression of function ${\rm{M}}\left( {{\rho _1},{\rho _2},{\rho _3},{\rho _4},{\rho _5}} \right)$. To deal with this problem, we go back to the integral expression of ${\rm{M}}\left( {{\rho _1},{\rho _2},{\rho _3},{\rho _4},{\rho _5}} \right)$ in (51). By definition, the McLaurin series expansion of ${\rm{M}}\left( {{\rho _1},{\rho _2},{\rho _3},{\rho _4},{\rho _5}} \right)$ can be expressed as
\begin{equation}
{\rm M} = {\left. {\rm M} \right|_{\gamma  = 0}} + {\left. {{{\rm M}^{(1)}}} \right|_{\gamma  = 0}}\gamma  + \frac{1}{2}{\left. {{{\rm M}^{(2)}}} \right|_{\gamma  = 0}}{\gamma ^2} + o\left( {{\gamma ^2}} \right)
\end{equation}
where ${{\rm M} ^{(n)}} = \frac{{{d^n}{\rm M} }}{{d{\gamma ^n}}}$ and $o\left( \delta  \right)$ indicates the higher order term of $\delta$. According to the result reported in [28, 0.410], ${\left. {{{\rm M} ^{(n)}}} \right|_{\gamma  = 0}}$ ($n=1,2$) can be derived directly from its integral expression (51), i.e.,
\begin{equation}
\begin{aligned}
&    {\left. {{{\rm M} ^{(1)}}} \right|_{\gamma  = 0}} = {\left[ {{{\left. {g\left( {\gamma ,z} \right)} \right|}_{z = \gamma }}} \right]_{\gamma  = 0}}\\
&{}  {\left. {{{\rm M} ^{(2)}}} \right|_{\gamma  = 0}} = {\left[ {{{\left. {\frac{{dg\left( {\gamma ,z} \right)}}{{d\gamma }}} \right|}_{z = \gamma }} + \frac{d}{{d\gamma }}\left( {{{\left. {g\left( {\gamma ,z} \right)} \right|}_{z = \gamma }}} \right)} \right]_{\gamma  = 0}}
\end{aligned}
\end{equation}
where $g\left( {\gamma ,z} \right)$ can be expressed as
\begin{equation}
g\left( {\gamma ,z} \right) = {\rho _1}\exp \left( { - {\lambda _2}\gamma } \right){\left( {\frac{1}{{{{\bar \gamma }_0}}}\gamma  + {\beta _{j,k}}} \right)^{ - {\rho _2}}}\frac{{\exp \left( { - {\Phi _2}z} \right)}}{{{\rho _1}z + {\rho _2}\gamma  + {\rho _3}}}\
\end{equation}
At last, following by some algebraic manipulation, one can arrived at the result in Lemma 3.

\section*{Appendix D: Proof of Theorem 3}
Substituting (11) into (22) and using [28, 9.211.4], $P_{{T_1}}^{{\cal E},\rm{LB}}$ can be expressed as in the form of (23) in Theorem 3, where the function ${\rm M}^{{\cal E}}(\rho_1,\rho_2,\rho_3,\rho_4,\rho_5)$ and ${\Lambda ^{\cal E}}(\rho_1,\rho_2)$ are expressed as
\begin{equation}
\begin{array}{ll}
{{\rm{M}}^{\cal E}}\left( {{\rho _1},{\rho _2},{\rho _3},{\rho _4},{\rho _5}} \right) = a\sqrt {\frac{b}{\pi }} \int\limits_0^\infty  {\frac{{\exp \left( { - b\gamma } \right)}}{{\sqrt \gamma  }}} {\rm{M}}\left( {{\rho _1},{\rho _2},{\rho _3},{\rho _4},{\rho _5}} \right)d\gamma\\
{\Lambda ^{\cal E}}\left( {{\rho _1},{\rho _2}} \right) = a\sqrt {\frac{b}{\pi }} \int\limits_0^\infty  {\frac{{\exp \left( { - b\gamma } \right)}}{{\sqrt \gamma  }}} \Lambda \left( {{\rho _1},{\rho _2}} \right)d\gamma
\end{array}
\end{equation}
Then the remaining task is to express (55) as in the form of Table I. Substituting (12) into the first line of (55) and replacing the lower incomplete gamma function with its series expansion [28, 8.354.1], i.e., ${\mathbb L}\left( {\alpha ,x} \right) = \sum_{n = 0}^\infty  {\frac{{{{\left( { - 1} \right)}^n}{x^{\alpha  + n}}}}{{n!\left( {\alpha  + n} \right)}}}$, ${\rm M}^{{\cal E}}(\rho_1,\rho_2,\rho_3,\rho_4,\rho_5)$ can be rewritten as
\begin{equation}
\begin{aligned}
&   {\rm M}^{\cal E}\left( {{\rho _1},{\rho _2},{\rho _3},{\rho _4},{\rho _5}} \right) = a\sqrt {\frac{b}{\pi }} \sum\limits_{l = 0}^\infty  {\sum\limits_{n = 0}^\infty  {\frac{{{{\left( { - 1} \right)}^{l + n}}\Phi _2^n{\rho _1}\rho _3^l}}{{n!\left( {l + n + 1} \right)}}} } \\
&{} \times \int\limits_0^\infty  {{\gamma ^{l + n + \frac{1}{2}}}} \frac{1}{{{{\left( {{\rho _4}\gamma  + {\rho _5}} \right)}^{l + 1}}}} \frac{\exp \left( { - \left( {{\lambda _2} + b} \right)\gamma } \right)}{{{{\left( {\frac{1}{{{{\bar \gamma }_0}}}\gamma  + {\beta _{j,k}}} \right)}^{{\rho _2}}}}}d\gamma
\end{aligned}
\end{equation}
Taking partial fraction on term $\frac{1}{{{{\left( {{\rho _4}\gamma  + {\rho _5}} \right)}^{l + 1}}}}\frac{1}{{{{\left( {\frac{1}{{{{\bar \gamma }_0}}}\gamma  + {\beta _{j,k}}} \right)}^{{\rho _2}}}}}$, i.e.,
\begin{equation}
\frac{1}{{{{\left( {{\rho _4}\gamma  + {\rho _5}} \right)}^{l + 1}}}}\frac{1}{{{{\left( {\frac{1}{{{{\bar \gamma }_0}}}\gamma  + {\beta _{j,k}}} \right)}^{{\rho_2}}}}}  = \sum\limits_{i = 1}^{l + 1} {{\mu _i}\frac{1}{{{{\left( {{\rho _4}\gamma  + {\rho _5}} \right)}^i}}} + } \sum\limits_{i = 1}^{{\rho _2}} {{\nu _i}} \frac{1}{{{{\left( {\frac{1}{{{{\bar \gamma }_0}}}\gamma  + {\beta _{j,k}}} \right)}^i}}}
\end{equation}
where $\mu _i$ and $\nu _i$ are given in (24), and employing equation [28, 9.211.4] on the resultant integrals, we can obtain the first row in Table I.

Similarly, substituting (13) into the second line of (55), applying partial fraction on term $\frac{1}{{\frac{1}{{{{\bar \gamma }_0}}}\gamma  + {\beta _{k,j}}}}\frac{1}{{{\rho _1}\gamma  + \frac{1}{{{\xi_{T_1,j}}}}}}$ and using [28, 9.211.4] on the resultant expression, one can yield the result of the fourth row in Table I.


\begin{thebibliography}{1}

\bibitem{IEEEhowto:Kim2011}
S. J. Kim, N. Devroye, P. Mitran, and V. Tarokh, ``Achievable Rate Regions and Performance Comparison of Half Duplex Protocols,'' \emph{IEEE Trans. Inf. Theory}, vol. 57, no. 10, pp. 6405-6418, Oct. 2011.

\bibitem{IEEEhowto:Upadhyay}
P. K. Upadhyay, and S. Prakriya, ``Performance of analog network coding with asymmetric traffic requirements,'' \emph{IEEE Commun. Lett.}, vol. 15, no. 6, pp. 647-649, Jun. 2011.

\bibitem{IEEEhowto:Kim_1}
Z. Yi, M. Ju, and I. Kim, ``Outage probability and optimum power allocation for analog network coding,'' \emph{IEEE Trans. Wireless Commun.}, vol. 10, no. 2, pp. 407-412, Feb. 2011.

\bibitem{IEEEhowto:Xiaochen2013IET}
X. Xia, K. Xu, W. Ma, and Y. Xu, ``On the design of relay selection strategy for two-way amplify-and-forward mobile relaying,'' to appear in \emph{IET Commun.}, 2013.

\bibitem{IEEEhowto:Kim_3}
Z. Yi, M. Ju, and I. Kim, ``Outage probability and optimum combining for time division broadcast protocol,'' \emph{IEEE Trans. Wireless Commun.}, vol. 10, no. 5, pp. 407-412, May 2011.

\bibitem{IEEEhowto:Lei2013}
X. Lei, L. Fan, D. S. Michalopoulos, P. Fan, and R. Q. Hu, ``Outage probability of TDBC protocol in multiuser two-way relay systems with Nakagami-$m$ fading,'' \emph{IEEE Commun. Lett.}, vol. 17, no. 3, pp. 487-490, Mar 2013.

\bibitem{IEEEhowto:YadavCL2013}
S. Yadav, and P. K. Upadhyay, ``Performance of three-Phase analog network coding with
relay selection in Nakagami-m fading,'' \emph{IEEE Commun. Lett.}, vol. 17, no. 8, pp. 1620-1623, Aug. 2013.

\bibitem{IEEEhowto:Yi2009TWC}
Z. Yi, and I. Kim ``An opportunistic based protocol for bidirectional cooperative networks,'' \emph{IEEE Wireless Commun.}, vol. 8, no. 9, pp. 4836-4847 , Sep. 2009.

\bibitem{IEEEhowto:Hoeher}
P.A. Hoeher, S. Badri-Hoeher, W. Xu, C. Krakowski, ``Single-antenna co-channel interference cancellation for TDMA cellular radio systems,'' \emph{IEEE Wireless Commun.}, vol. 12, no. 2, pp. 30-37 , Apr. 2005.

\bibitem{IEEEhowto:Ikki}
S. S. Ikki, and S. Aissa, ``Performance analysis of two-way amplify-and-forward relaying in the presence of co-channel interferences,'' \emph{IEEE Trans. Commun.}, vol. 60, no. 4, pp. 933-939, Apr. 2012.

\bibitem{IEEEhowto:Liang2012}
X. Liang, S. Jin, W. Wang, X. Gao, and K. wong, ``Outage probability of amplify-and-forward two-way relaly interference-limited system,'' \emph{IEEE Trans. Veh. Tech.}, vol. 61, no. 7, pp. 3038-3049, Sep. 2012.

\bibitem{IEEEhowto:Anup2013}
A. K. Mandpura, S. Prakriya, and R. K. Mallik, ``Outage probability of amplify-and-forward two-way cooperative systems in presence of multiple co-channel interferers,''  \emph{IEEE NCC 2013}, Feb. 2013.

\bibitem{IEEEhowto:Soleimani2013TCOM}
E. Soleimani-Nasab, M. Matthaiou, M. Ardebilipour, and G. K. Karagiannidis, ``Two-way AF relaying in the presence of co-channel interference,'' \emph{IEEE Trans. Commun.}, vol. 61, no. 8, pp. 3156-3169, Aug. 2013.

\bibitem{IEEEhowto:Xiaochen}
X. Xia, Y. Xu, K. Xu, D. Zhang, and N. Li, ``Outage Performance of AF-based Time Division Broadcasting Protocol in the Presence of Co-channel Interference,'' \emph{IEEE WCNC 2013}, Shanghai, China, Apr. 2013.

\bibitem{IEEEhowto:Ding}
H. Ding, J. Ge and D. B. Costa, ``Two birds with one stone: exploiting direct links for multiuser two-way relaying system,'' \emph{IEEE Trans. Wireless Commun.}, vol. 11, pp. 54-59, Jan. 2012.

\bibitem{IEEEhowto:Lee}
H. Yu, I. Lee, and G. L. St${{\rm{\ddot {u}}}}$ber ``Outage probability of decode-and-forward cooperative relaying systems with co-channel interference,'' \emph{IEEE Trans. Wireless Commun.}, vol. 11, no. 1, pp. 266-274, Jan. 2012.

\bibitem{IEEEhowto:Shah2000}
A. Shah, and A. M. Haimovich, ``Performance Analysis of Maximal Ratio Combining and Comparison with Optimum Combining for Mobile Radio Communications with Cochannel Interference,'' \emph{IEEE Trans. Veh. Tech.}, vol. 49, no. 4, pp. 1453-1463, Jul. 2000.

\bibitem{IEEEhowto:Chayawan2002}
C. Chayawan, and V. A. Aalo, ``On the outage probability of optimum combining and maximal ratio combining schemes in an interference-Limited rice fading channel,'' \emph{IEEE Trans. Commun.}, vol. 50, no. 4, pp. 532-535, Apr. 2002.

\bibitem{IEEEhowto:Khuong}
H. V. Khuong and H. Kong, ``General expression for pdf of a sum of independent exponential random variables,'' \emph{IEEE Commun. Lett.}, vol. 10, no. 3, Mar 2006.

\bibitem{IEEEhowto:Ikki2013}
S. S. Ikki, P. Ubaidulla, and S. A{\"i}ssa, ``Performance study and optimization of cooperative diversity networks with co-channel interference'' accepted for publication in \emph{IEEE Trans. Wireless Commun.}, 2013.

\bibitem{IEEEhowto:Wang}
Z. Wang, and G. B. Giannakis, ``A simple and general parameterization quantifying performance in fading channels,'' \emph{IEEE Trans. Commun.}, vol. 51, no. 8, pp. 1389-1398, Aug. 2003.

\bibitem{IEEEhowto:Suraweera2011ICC}
H. A. Suraweera., D. S. Michalopoulos, R. Schober, G. K. Karagiannidis, and A. Nallanathan, ``Fixed gain amplify-and-forward relaying with co-channel interference,'' \emph{IEEE ICC 2011}, Kyoto, Japan, Jun 2011.

\bibitem{IEEEhowto:Salhab2013CL}
A. M. Salhab, F. Al-Qahtani, S. A. Zummo, and H. Alnuweiri, ``Outage Analysis of $N$th-Best DF Relay Systems in the Presence of CCI over Rayleigh Fading Channels,'' \emph{IEEE Commun. Lett.}, vol. 17, no. 4, pp. 19-22, Apr. 2013.

\bibitem{IEEEhowto:Louie2010}
R. Louie, Y. Li, and B. Vucetic, ``Practical physical layer network coding for two-way relay channels: performance analysis and comparison,'' \emph{IEEE Trans. Wireless Commun.}, vol. 9, no. 2, pp. 764-777, Feb. 2010.

\bibitem{IEEEhowto:Chen2004}
Y. Chen and C. Tellambura, ``Distribution function of selection combiner output in equally correlated Rayleigh, Rician, and Nakagami-m fading channels,'' \emph{IEEE Trans. Commun.}, vol. 52, no. 11, pp. 1948-1956, Nov. 2004.

\bibitem{IEEEhowto:Al-Qahtani}
F. S. Al-Qahtani, J. Yang, R. M. Radaydeh, and H. Alnuweiri, ``On the capacity of two-hop AF relaying in the presence of interference under Nakagami-m fading,'' \emph{IEEE Commun. Lett.}, vol. 17, no. 1, pp. 19-22, Jan. 2013.

\bibitem{IEEEhowto:Rappaport}
T. S. Rappaport, \emph{Wireless Communications: Principles and Practice.} Prentice Hall, 2002.


\bibitem{IEEEhowto:20}
I. S. Gradshteyn and I. M. Ryzhik, \emph{Table of Integrals, Series and Products}, 7th edition. Academic Press, 2007.

\bibitem{IEEEhowto:Yang2013TC}
L. Yang, K. Qaraqe, E. Serpedin, and M.louini, ``Performance analysis of amplify-and-forward two-way relaying with co-channel interference and channel estimation error,'' accepted for publication in \emph{IEEE Trans. Commun.}, 2013.

\bibitem{IEEEhowto:Pun2004GC}
M. Pun, S. Tsai and C. J. Kuo, ``Joint maximum likelihood estimation of carrier frequency offset and channel in uplink OFDMA systems'' \emph{IEEE Globecom 2004}, pp. 3748-3752, Dec. 2004.

\bibitem{IEEEhowto:Ziskind1988}
I. Ziskind and M. Wax, ``Maximum likelohood localization of multiple sources by alternating projection'', \emph{IEEE Trans. Acoust. Speech, Signal Processing}, vol. 36, no. 10, Oct. 1988.

\end{thebibliography}
\end{document}